\documentclass[runningheads]{llncs}
\usepackage[T1]{fontenc}
\usepackage{booktabs}
\usepackage{newtxtext} 
\usepackage{diagbox}%
\usepackage[varvw]{newtxmath}       

\usepackage{hyperref}
\usepackage{graphicx}
\usepackage{xspace}
\newcommand{\cipher}{\textsc{Lilliput}\xspace}

\newtheorem{Proposition}{Proposition}
\begin{document}
\title{Differential Fault Analysis of \cipher under Random-Location Nibble Faults}
%
%

\author{Peipei Xie\inst{1} \and
 Siwei Chen\inst{1}\thanks{Corresponding author.} \and
 Zejun Xiang\inst{1} \and
 Shasha Zhang\inst{1} \and
 Xiangyong Zeng\inst{2}}
 \authorrunning{P. Xie et al.}
 \institute{
 School of Cyber Science and Technology,
 Hubei University, Wuhan, China
 \and
 Faculty of Mathematics and Statistics,
 Hubei Key Laboratory of Applied Mathematics,
 Hubei University, Wuhan, China\\
 \email{chensiwei\_hubu@163.com}\\
 \email{xiepeipei@stu.hubu.edu.cn}\\
 \email{xiangzejun@hubu.edu.cn}\\
 \email{amushasha@163.com}\\
 \email{xzeng@hubu.edu.cn}
 }

%
\maketitle              
\begin{abstract}
\cipher, a lightweight block cipher proposed by Berger et al. in \textit{IEEE Transactions
on Computers (2015)}, is built upon a Type-2 Extended Generalized Feistel Network (EGFN). It has been extensively studied against conventional cryptanalytic techniques, including differential, linear, integral, and boomerang cryptanalysis.
However, its security against fault-based attacks remains largely unexplored, neither theoretical analyses nor practical/physical fault-based security results have been reported. This paper presents a theoretical differential fault analysis (DFA) evaluation of \cipher under a fixed-timing, full-branch random-location nibble fault model, in which a nonzero nibble fault is injected at the input to round 27, with its location uniformly distributed across all sixteen state branches and unknown to the adversary.

The main challenge is to convert faulty ciphertexts with unknown injection locations into usable key-recovery constraints. We analyze the fault propagation induced by the EGFN structure of \cipher and derive ciphertext-difference conditions for identifying the injected branch. 
A DDT-weighted calculation gives an overall branch-identification accuracy of approximately $99.9980\%$ under independent uniform sampling of plaintexts, nonzero fault values, and fault locations; the observed accuracy is $99.9983\%$ in $2^{20}$ random fault simulations. Once the fault branch is determined, we classify the corresponding propagation patterns according to whether the injected fault value and the intermediate S-box output difference can be uniquely determined. For each case, we derive DDT-based constraints on the last-round and penultimate-round subkeys and combine multiple faulty ciphertexts by candidate-set intersection. Simulation experiments over $2^{15}$ trials show that the attack reaches key-recovery success rates of over $90\%$, $95\%$, and $99\%$ with 32, 36, and 46 faulty ciphertexts, respectively. We state that the reported results are obtained through simulation-based evaluation rather than physical fault-injection experiments on actual devices, but they provide useful insights for practical DFA on \cipher in physical settings.
\end{abstract}
\section{Introduction}\label{sec:1}
Lightweight cryptographic implementations operate under tight area, memory, and power constraints. These budgets can also restrict the use of redundancy, fault detection, and other physical protections. Evaluating a lightweight cipher therefore requires attention not only to fault-free cryptanalysis but also to the key information exposed when an internal computation is disturbed.

The block cipher \cipher~\cite{berger2015extended} uses an Extended Generalized Feistel Network (EGFN) with a matrix-based representation~\cite{berger2013extended} to provide diffusion with a compact hardware footprint. Its security has been studied through (impossible) differential attacks~\cite{marriere2018differential,nachef2016improved,sasaki2016impossible,sasaki2017tight}, linear and integral attacks~\cite{sasaki2017tight}, and boomerang attacks~\cite{wu2026improved}. These analyses concern fault-free encryption. Transient faults introduce a different source of information: the differences between correct and faulty computations may reveal relations involving the secret key. For a multi-branch design such as \cipher, the available relations depend on where the fault enters the state and how it propagates through the subsequent rounds.

\noindent\textbf{Related work.}
Differential fault analysis (DFA)~\cite{biham1997differential} exploits correct and faulty ciphertexts generated from the same plaintext under the same key. It has been applied to DES~\cite{biham1997differential}, AES~\cite{ali2013differential}, SKINNY~\cite{yu2023automatic}, QARMAv2~\cite{sahoo2025unleashing}, Kreyvium~\cite{roy2020differential}, and Elisabeth~\cite{wang2025differential}. The assumptions on fault timing, granularity, and location determine which internal differences can be used to form key-recovery equations. A specified location fixes the starting point of the propagation path; an unknown location requires the analysis to account for several possible paths.

Earlier studies have addressed this uncertainty through ciphertext-difference patterns. Jeong et al.~\cite{jeong2013improved} identify random nibble faults in LBlock when both the location and value are unknown; each of their two models allows eight positions in one half of the state at a specified round. Sahoo et al.~\cite{sahoo2025unleashing} use nibble-level propagation patterns to identify unknown fault positions and recover keys in QARMAv2. At the structural level, Le Bouder et al.~\cite{lebouder2014fault} analyze single-bit fault propagation in generalized Feistel networks, relating the injection round and branch to subkey exposure and fault requirements. These results show how fault localization supports key recovery, while also establishing that different branches can supply different amounts of key information. Both aspects matter when the attacker cannot select a favorable injection branch.

\vspace{5pt}
\noindent\textbf{Motivation.}
We investigate how these two aspects interact in \cipher when the possible fault locations span the entire state. We fix the injection point at the input to round~27 and allow a nonzero nibble fault in any of the sixteen branches. Both the affected branch and the fault value are unknown to the attacker. This setting retains round-level timing control but removes the need to select a particular branch in advance.

Two difficulties arise in converting the resulting ciphertext differences into subkey constraints. First, some branch-identification conditions overlap, so satisfying a condition does not always determine the true injection branch. A wrong assignment may produce incompatible constraints and eliminate the correct subkey during candidate filtering. The reliability of the location rule must therefore be quantified using the probabilities of the underlying S-box differential transitions. Second, even a correct branch assignment need not determine the injected fault value or an intermediate S-box output difference. Such unknown values must remain as candidates when the subkey constraints are formed. This motivates analyzing location ambiguity and fault-value uncertainty together, rather than treating branch identification as a sufficient step for key recovery.

\vspace{5pt}
\begin{samepage}
\noindent\textbf{Our contributions.}
We derive ciphertext-difference conditions for all sixteen injection branches and characterize their ambiguities. Under the resulting decision rule, misclassification can occur only when a fault at branch~7 satisfies the conditions for branches~0, 2, 3, or~15. We calculate the probabilities of these events using the S-box differential distribution table (DDT), obtaining an overall identification accuracy of approximately $99.9980\%$ under the sampling assumptions in Section~\ref{sec:3}.

\end{samepage}
For each assigned branch, we then distinguish the cases in which both the fault value and the intermediate difference are known, only the fault value is known, or both remain uncertain. We derive the corresponding constraints on the last two round subkeys, take unions over admissible fault-value and intermediate-difference candidates, and intersect the resulting key candidates across faulty ciphertexts. The recovered subkey information is linked to the master key through the key schedule.

Simulations measure both branch-identification accuracy and the number of faults needed for subkey recovery. The location rule achieves $99.9983\%$ accuracy over $2^{20}$ random fault injections, consistent with the DDT-weighted calculation. In $2^{15}$ recovery experiments, successful recovery of both $RK^{29}$ and $RK^{28}$ requires between 2 and 134 faulty ciphertexts. Leaving the bottleneck nibble $RK^{28}_3$ to the final exhaustive search reduces this observed range to 2--76. In the latter setting, fixed budgets of 32, 36, and 46 faulty ciphertexts give target-subkey recovery rates of $91.15\%$, $95.24\%$, and $99.02\%$, respectively. We emphasize that the reported results are obtained solely from theoretical analysis and simulation-based evaluation under the stated nibble-level fault model, and have not been validated through physical fault-injection experiments on actual devices.

\vspace{3pt}
\noindent\textbf{Organization.}
Section~\ref{sect:2} introduces \cipher and the fault model. Section~\ref{sec:3} derives the branch-identification rule and its error probability. Section~\ref{sec:4} develops the fault-value and subkey constraints, and Section~\ref{sec:5} reports the experiments. Section~\ref{sec:6} concludes.

\section{Preliminaries}\label{sect:2}
We begin this section by explaining the notations used in this paper, then briefly introduce the \cipher block cipher. The notations used throughout this paper are listed in Table~\ref{tab:notation}.
Additionally, we denote by $\tilde{\mathrm{x}}$ the faulty state corresponding to a correct state $\mathrm{x}$; for example, if $C$ is the ciphertext generated by encrypting a plaintext $P$, then $\tilde{C}$ represents the faulty ciphertext obtained after injecting a fault into the internal state. Moreover, for a given differential transition $(\alpha, \beta)$ of \cipher's S-box $S$, we define $\mathtt{IN}(\alpha,\beta)$ as follows:
\begin{equation*}
	\mathtt{IN}\left (\alpha,\beta \right ) =\left \{ x\mid x\in \mathbb{F}_2^4,S(x)\oplus S(x\oplus \alpha)=\beta    \right \}. 
\end{equation*}
Furthermore, we use $\mathtt{DDT}(\alpha, \beta)$ to denote the number of elements in $\mathtt{IN}\left (\alpha,\beta \right )$, i.e., $\mathtt{DDT}(\alpha, \beta) = |\mathtt{IN}\left (\alpha,\beta \right )|.$

\subsection{Specification of \cipher}
\cipher is a lightweight block cipher designed by Berger et al.~\cite{berger2015extended} for resource-constrained applications such as IoT devices. It features a 64-bit block size, an 80-bit key, and adopts a Type-2 EGFN structure in the round function over 30 encryption rounds. The overall encryption process and the detailed round function are depicted in Fig.~\ref{fig_1} and~\ref{fig_2}, respectively.
\begin{table}[!htp]
	\caption{Notations throughout this paper\label{tab:notation}}
	\centering
	\renewcommand{\arraystretch}{1.3}%
	\setlength{\tabcolsep}{4pt}%
	\begin{tabular}{ll}
		\toprule
		Notation & Description \\
		\midrule
		$\mathbb{F}_2$ & The finite field with two elements, 0 and 1.\\
		$\mathbb{F}_2^n$ & The $n$-dimensional vector space over $\mathbb{F}_2$.\\
		$P,C\in\mathbb{F}_2^{64}$ & The 64-bit plaintext and ciphertext, respectively.\\
		$C_i\in\mathbb{F}_2^4$ & The $i$-th ($0\leq i\leq 15$) nibble of $C$.\\
		$X^r\in\mathbb{F}_2^{64}$& The input state of the $r$-th encryption round. \\
		$Y^r\in\mathbb{F}_2^{64}$& The state that becomes $X^{r+1}$ after the permutation \\& in the $r$-th encryption round. \\
		$X^r_i, Y^r_i\in\mathbb{F}_2^{4}$ & The $i$-th ($0\leq i\leq 15$) nibble of $X^r$ and $Y^r$.\\
		$RK^r\in\mathbb{F}_2^{32}$ & The 32-bit subkey in the $r$-th encryption round.\\
		$RK^r_i\in\mathbb{F}_2^4$ & The $i$-th nibble ($0\leq i\leq 7$) of $RK^r$.\\
		$\Delta_{\mathrm{x}}$ & The difference corresponding to the nibble state $\mathrm{x}$.\\
		$\oplus$ & Bitwise XOR operation. \\
		$\parallel$ & Concatenation operation of binary strings or vectors.\\
		$S^{r}_{i}$ & The $i$-th S-box in the $r$-th encryption round. \\
		$\alpha^{r}_{i}, \beta^{r}_{i}$ & The input and output differences of the S-box $S^{r}_{i}$. \\
		\bottomrule
	\end{tabular}
\end{table}
\begin{figure}[!h]
    \centering
    \begin{minipage}{0.4\textwidth}
        \centering
        \includegraphics[width=\linewidth]{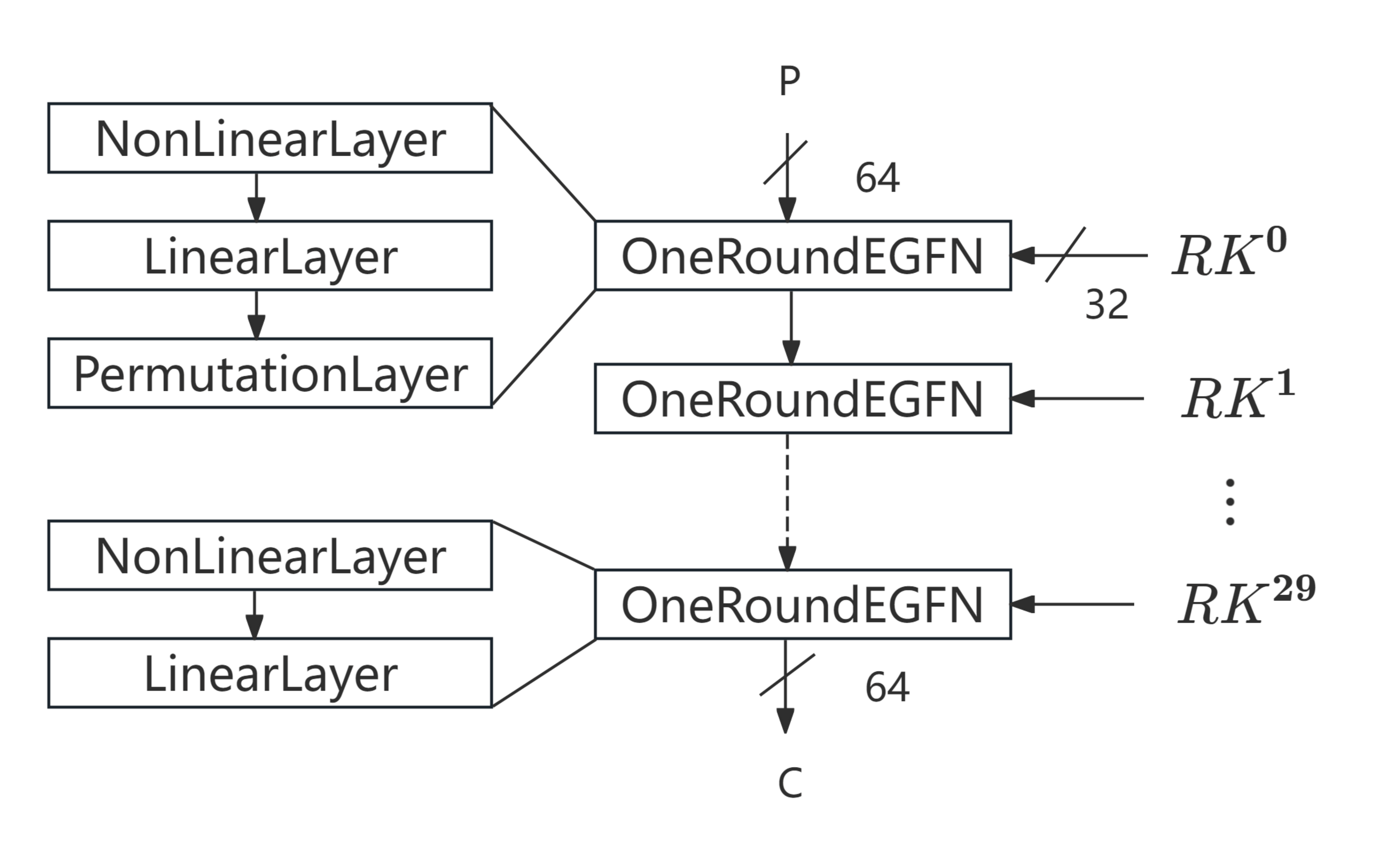}
        \caption{Overview of \cipher's encryption process.}
        \label{fig_1}
    \end{minipage}
    \hfill
    \begin{minipage}{0.56\textwidth}
        \centering
        \includegraphics[width=\linewidth]{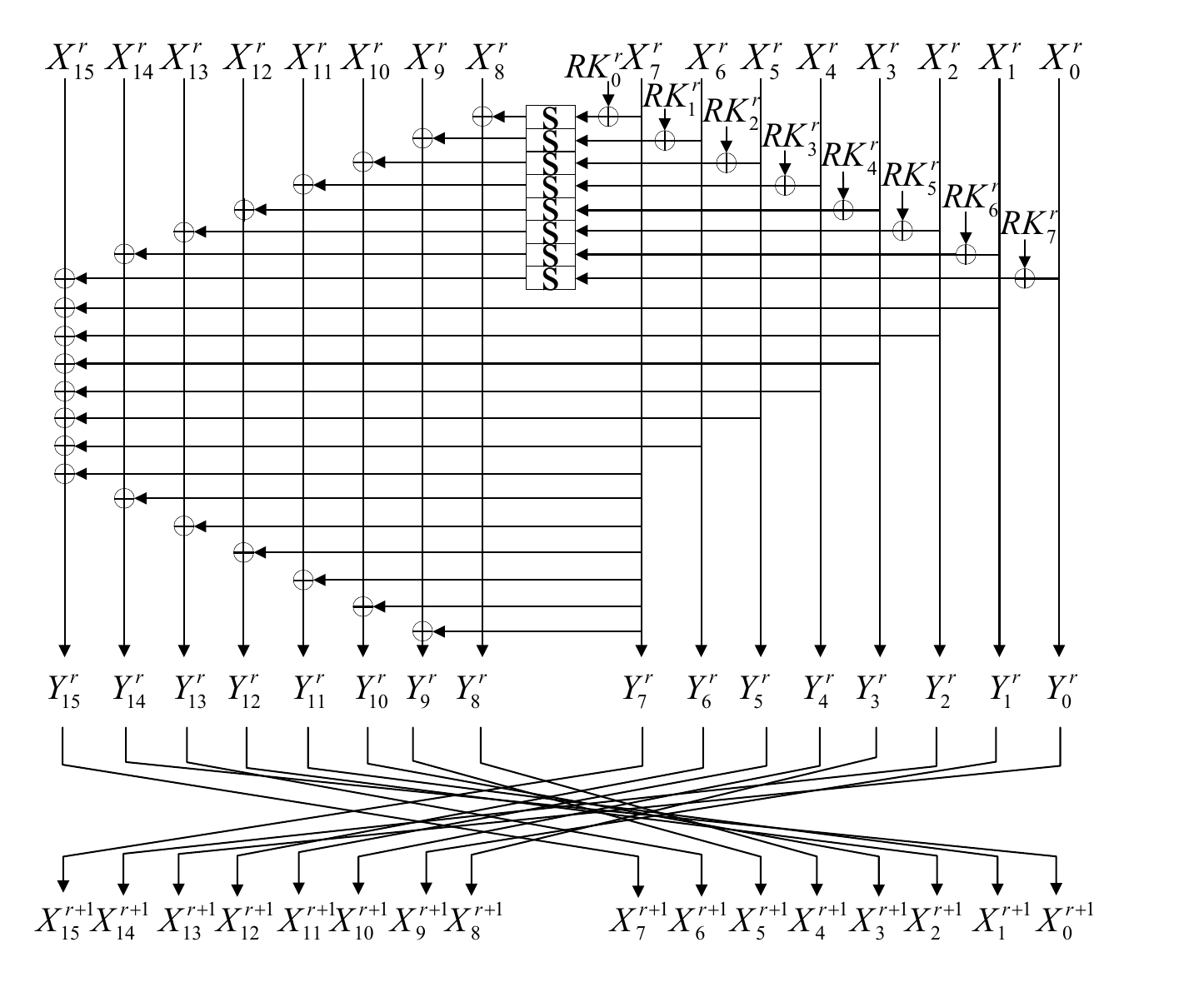}
        \caption{The $r$-th round encryption of \cipher.}
        \label{fig_2}
    \end{minipage}
\end{figure}

\cipher's round function, denoted by \textsf{OneRoundEGFN}, operates at the nibble level and consists of three layers: \textsf{NonLinearLayer}, \textsf{LinearLayer}, and \textsf{PermutationLayer}. Specifically, the final (i.e., $29$-th) encryption round omits \textsf{PermutationLayer}. In the $r$-th encryption round($0\leq r\leq 28$), the 64-bit input state $X^r$ is divided into 16 nibbles as $X^r = X^r_{15}\parallel \cdots \parallel X^r_{0}$, where $X^r_{15}$ is the most significant nibble (as shown in Fig.~\ref{fig_2}). We describe the three layers of the round function as follows:
\begin{itemize}
	\item \textsf{NonLinearLayer.} This layer first operates on $X^r_i$ and the subkey $RK^r$ to generate the intermediate state as $F_i = S(X^r_{7-i}\oplus RK^r_i)$ for $0\leq i\leq 7$, and then updates the eight left branches by XORing $F_i$ with $X^r_{i+8}$. The truth table of the S-box is shown in Table~\ref{tab:sbox}. For convenience, we denote the S-box which operates on $X^r_{7-i}\oplus RK^r_i$ by $S^r_i$, as already defined in Table~\ref{tab:notation}.
	\begin{table}[!h]
		\caption{The truth table of \cipher's S-box\label{tab:sbox}}
		\centering
		\setlength{\tabcolsep}{3pt} 
		\begin{tabular}{|c|c|c|c|c|c|c|c|c|c|c|c|c|c|c|c|c|}
			\hline
			$x$ & 0 & 1 & 2 & 3 & 4 & 5 & 6 & 7 & 8 & 9 & A & B & C & D & E & F\\
			\hline
			$S(x)$ & 4 & 8 & 7 & 1 & 9 & 3 & 2 & E & 0 & B & 6 & F & A & 5 & D  & C\\
			\hline
		\end{tabular}
	\end{table}
	\item \textsf{LinearLayer.} In this step, the seven leftmost branches are further updated by XORing them with $X^r_i$ ($1\leq i\leq 7$), respectively. Consequently, $Y^r$ is obtained.
	\begin{table}[!h]
		\caption{Permutation of \cipher\label{tab:pi}}
		\centering
		\setlength{\tabcolsep}{3pt} 
		\begin{tabular}{|c|c|c|c|c|c|c|c|c|c|c|c|c|c|c|c|c|}
			\hline
			$j$ & 0 & 1 & 2 & 3 & 4 & 5 & 6 & 7 & 8 & 9 & 10 & 11 & 12 & 13 & 14 & 15\\
			\hline
			$\pi(j)$ & 13 & 9 & 14 & 8 & 10 & 11 & 12 & 15 & 4 & 5 & 3 & 1 & 2 & 6 & 0 & 7\\
			\hline
		\end{tabular}
	\end{table}
	\item \textsf{PermutationLayer.} This layer permutes the position of $Y^r_i$ to obtain $X^{r+1}_i$ as $X^{r+1}_{\pi(i)} = Y^r_i$ for $0\leq i\leq 15$, where $\pi$ is defined in Table~\ref{tab:pi}.
\end{itemize}
The key schedule of \cipher utilizes an 80-bit Linear Finite State Machine (LFSM), which is initialized with the 80-bit master key $MK$. The first subkey $RK^0$ is directly extracted from the LFSM by the \textsf{ExtractRoundKey} function. Subsequently, each time the LFSM updates via the \textsf{RoundFnLFSM} function, a new subkey is generated. In total, the LFSM is updated 29 times. More details are illustrated in Supplementary Material~\ref{sm:key_schedule}.

\subsection{DFA Model and Choice of Injection Timing for \cipher}
\label{subsec:dfa_model}
Differential fault analysis (DFA) is a cryptanalytic technique that exploits the differences between correct and faulty computations. In a typical DFA setting, the attacker induces a transient fault during the execution of a cipher and obtains, for the same plaintext and secret key, both the correct ciphertext $C$ and the faulty ciphertext $\widetilde{C}$. By analyzing the ciphertext difference $\Delta_C=C\oplus \widetilde{C}$ together with the cipher structure, key-dependent constraints can be derived and used to recover subkey information.

In this paper, we consider a fixed-timing random-location nibble fault model for \cipher. The injection timing is fixed to round~27, i.e., the third round from the end. For each fault injection, a nonzero nibble fault is injected into one of the sixteen state branches $X^{27}_0,\ldots,X^{27}_{15}$, where the affected branch is randomly distributed and unknown to the attacker. The fault value is also unknown. The attacker is assumed to know the specification of \cipher, the fault model, and the injection round, and can obtain correct/faulty ciphertext pairs generated from the same plaintext under the same secret key. However, the attacker does not know the master key, the actual fault branch, the fault value, or any intermediate state. 

The choice of round~27 is motivated by a trade-off between fault-location distinguishability and exploitable key information. For \cipher, the 8 branches in the left half and the 8 branches in the right half achieve full diffusion after 5 rounds and 4 rounds, respectively.
If the fault is injected too early, for example in round~26, the difference has more time to diffuse before reaching the ciphertext. Under the random-location model, this makes the originating branch difficult to distinguish from ciphertext differences alone and weakens the subsequent construction of branch-specific DFA equations. Conversely, if the fault is injected too late, for example in round~28, the fault branch and fault value may be easier to infer, but the propagation involves too few key-dependent S-boxes to provide sufficiently strong constraints for efficient master-key recovery. Round~27 provides a useful balance: the ciphertext difference still retains enough branch-dependent structure for location identification, while the propagation activates enough nonlinear components in the last two rounds to support DDT-based subkey recovery.

We do not claim that round-level timing control is trivial in all implementations; rather, we use it as an explicit algorithmic fault abstraction to evaluate the last-round fault propagation of \cipher.
\section{Fault Location Identification From Ciphertext Difference Conditions}\label{sec:3}
Faults injected into different branches follow different propagation paths, so the construction of subkey constraints requires a decision about the fault location. Let $L\in\{0,\ldots,15\}$ denote the actual injection branch. The round function gives the necessary ciphertext-difference conditions $\mathcal{C}_i$ in Table~\ref{tab:fault_branch_conditions}: a fault at $X_i^{27}$ always satisfies $\mathcal{C}_i$. We use the following decision rule: if $\mathcal{C}_i$ holds for some $i\ne7$, set $\widehat L=i$; otherwise, set $\widehat L=7$. The conditions with $i\ne7$ are mutually exclusive. We first determine which assignments are unambiguous and then calculate the probability of an incorrect assignment.

\begin{table}[!htp]
\centering
\caption{Conditions on ciphertext difference for different fault injection positions.}
\label{tab:fault_branch_conditions}
\renewcommand{\arraystretch}{1.2}
\setlength{\tabcolsep}{2pt}
\resizebox{\textwidth}{!}{
\begin{tabular}{ll}
\toprule
Position & Conditions ($\mathcal{C}_i$) \\
\midrule

$X^{27}_0$ &
$\Delta_{C_i}=\Delta_{C_j}\neq0$ for $i,j\in\{0,1,2,3,5,7\}$; 
$\Delta_{C_i}\neq\Delta_{C_6}$ for $i\in\{0,1,2,3,5,7\}$;
$\Delta_{C_4}\neq 0$. \\

$X^{27}_1$ &
$\Delta_{C_5}=0$; 
$\Delta_{C_4}\neq 0$;
$\Delta_{C_i}=\Delta_{C_j}\neq0$ for $i,j\in\{0,1,2,3,6\}$; 
$\Delta_{C_i}\neq\Delta_{C_7}=\Delta_{C_{10}}$ for $i\in\{0,1,2,3,6\}$. \\

$X^{27}_2$ &
$\Delta_{C_0}=0$;
 $\Delta_{C_4}\neq 0$;
$\Delta_{C_i}=\Delta_{C_j}\neq0$ for $i,j\in\{1,2,3,6\}$; 
$\Delta_{C_i}\neq\Delta_{C_5},\Delta_{C_7}$ for $i\in\{1,2,3,6\}$. \\

$X^{27}_3$ &
$\Delta_{C_i}=\Delta_{C_j}\neq0$ for $i,j\in\{0,1,2,3,5\}$; 
$\Delta_{C_i}\neq\Delta_{C_j}$ for $i\in\{0,1,2,3,5\}$ and $j\in\{4,6,7\}$. \\

$X^{27}_4$ &
$\Delta_{C_3}=0$;
$\Delta_{C_4}\neq 0$;
$\Delta_{C_7}=\Delta_{C_{12}}$; 
$\Delta_{C_i}=\Delta_{C_j}\neq0$ for $i,j\in\{1,2,5,6\}$; 
$\Delta_{C_i}\neq\Delta_{C_j}$ for $i\in\{1,2,5,6\}$ and $j\in\{0,7\}$. \\

$X^{27}_5$ &
$\Delta_{C_1}=0$;
$\Delta_{C_4}\neq 0$; 
$\Delta_{C_7}=\Delta_{C_{14}}$; 
$\Delta_{C_i}=\Delta_{C_j}\neq0$ for $i,j\in\{0,3,5,6\}$; 
$\Delta_{C_i}\neq\Delta_{C_j}$ for $i\in\{0,3,5,6\}$ and $j\in\{2,7\}$. \\

$X^{27}_6$ &
$\Delta_{C_2}=0$; 
$\Delta_{C_4}\neq 0$;
$\Delta_{C_7}=\Delta_{C_{13}}$; 
$\Delta_{C_i}=\Delta_{C_j}\neq0$ for $i,j\in\{0,1,5,6\}$; 
$\Delta_{C_i}\neq\Delta_{C_j}$ for $i\in\{0,1,5,6\}$ and $j\in\{3,7\}$. \\

$X^{27}_7$ & $\Delta_{C_4}\neq 0$\\

$X^{27}_8$ &
$\Delta_{C_i}=0$ for $i\in\{0,2,3,4,5,6\}$; 
$\Delta_{C_1}\neq0$ and $\Delta_{C_7}\neq0$. \\

$X^{27}_9$ &
$\Delta_{C_i}=0$ for $i\in\{0,1,2,4,5,6\}$; 
$\Delta_{C_3}\neq0$ and $\Delta_{C_7}\neq0$. \\

$X^{27}_{10}$ &
$\Delta_{C_i}=0$ for $i\in\{0,1,3,4,5,6\}$; 
$\Delta_{C_2}\neq0$ and $\Delta_{C_7}\neq0$. \\

$X^{27}_{11}$ &
$\Delta_{C_i}=0$ for $i\in\{1,2,3,4,5,6\}$; 
$\Delta_{C_0}\neq0$ and $\Delta_{C_7}\neq0$. \\

$X^{27}_{12}$ &
$\Delta_{C_i}=0$ for $i\in\{0,1,2,3,4,5\}$; 
$\Delta_{C_6}\neq0$ and $\Delta_{C_7}\neq0$. \\

$X^{27}_{13}$ &
$\Delta_{C_i}=0$ for $i\in\{0,1,2,3,4,6\}$; 
$\Delta_{C_5}\neq0$ and $\Delta_{C_7}\neq0$. \\

$X^{27}_{14}$ &
$\Delta_{C_i}=0$ for $i\in\{0,1,2,3,4,5,6\}$; 
$\Delta_{C_7}\neq0$. \\

$X^{27}_{15}$ &
$\Delta_{C_i}=\Delta_{C_j}\neq0$ for $i,j\in\{0,1,2,3,5,6,7\}$; 
$\Delta_{C_4}\neq0$. \\

\bottomrule
\end{tabular}
}
\end{table}

\begin{Proposition}\label{prop1}
Let $\mathcal{J}=\{1,4,5,6,8,9,10,11,12,13,14\}$ and $\mathcal{A}=\{0,2,3,15\}$. If $\mathcal{C}_i$ holds for $i\in\mathcal{J}$, then $L=i$. If $\mathcal{C}_i$ holds for $i\in\mathcal{A}$, then $L\in\{i,7\}$. If none of the conditions with $i\ne7$ holds, then $L=7$.
\end{Proposition}
\begin{proof}
The zero/nonzero and equality/inequality requirements in Table~\ref{tab:fault_branch_conditions} make the conditions $\mathcal{C}_i$, $i\ne7$, pairwise incompatible. Hence, if $\mathcal{C}_i$ holds for $i\ne7$, the only possible injection branches are $i$ and~7. For $i\in\{8,\ldots,14\}$, the requirement $\Delta_{C_4}=0$ contradicts $\mathcal{C}_7$.

For a fault at branch~7, let $e\ne0$ be the injected difference and $e'$ the output difference of $S^{27}_0$. Since the S-box is bijective, $e'\ne0$. The resulting ciphertext differences satisfy the relations
in Equation~\eqref{equ:ciphertext_difference}, which excludes $\mathcal{C}_1,\mathcal{C}_4,\mathcal{C}_5,\mathcal{C}_6$. Specifically, if $\Delta_{C_5}=0$, the corresponding last-round S-box has zero output difference, so $\Delta_{C_{10}}\oplus\Delta_{C_7}=e'\ne0$, contradicting $\mathcal{C}_1$. Similarly, $\Delta_{C_3}=0$, $\Delta_{C_1}=0$, and $\Delta_{C_2}=0$ imply, respectively,
\[
\Delta_{C_{12}}\oplus\Delta_{C_7}=e,\qquad
\Delta_{C_{14}}\oplus\Delta_{C_7}=e,\qquad
\Delta_{C_{13}}\oplus\Delta_{C_7}=e.
\]
These nonzero differences contradict the additional equalities in $\mathcal{C}_4,\mathcal{C}_5,\mathcal{C}_6$. Finally, every fault satisfies its own branch condition, so the absence of all conditions with $i\ne7$ forces $L=7$.
\end{proof}

Thus, an identification error occurs only when a fault at branch~7 satisfies one of the four conditions indexed by $\mathcal{A}$. The rule assigns the corresponding $i$ in these ambiguous cases; its accuracy depends on the distribution of faults and plaintexts.

\vspace{5pt}
\noindent\textbf{Probability model.}
Fix an arbitrary master key. For the following single-fault calculation, sample the plaintext uniformly from $\mathbb{F}_2^{64}$, the nonzero XOR fault $e$ uniformly from $\mathbb{F}_2^4\setminus\{0\}$, and the location $L$ uniformly from the sixteen branches, independently. The result is an average under this sampling model, not a worst-case guarantee for a fixed plaintext. In particular, it does not assume independent identification errors across repeated faults on the same plaintext.

\begin{Proposition}\label{prop2}
Under the probability model above, the branch-identification rule has overall accuracy
\[
\Pr[\widehat L=L]
=1-\frac{644513}{32212254720}
\approx99.997999\%.
\]
\end{Proposition}
\begin{proof}
Write $P_\alpha(\beta)=\mathtt{DDT}[\alpha][\beta]/16$ and $q_i=\Pr[\mathcal{C}_i\mid L=7]$. The nonzero entries of the S-box DDT are not all equal: entries~2 and~4 correspond to probabilities $1/8$ and $1/4$. Consequently, feasible difference patterns must be weighted by their DDT probabilities rather than counted as equiprobable outcomes.

For a fixed key, the first 27 rounds form a permutation, so a uniform plaintext gives a uniform state $X^{27}$. Conditioned on its right half $(X^{27}_7,\ldots,X^{27}_0)$, each right-half nibble entering round~28 is a distinct uniform left-half nibble XORed with a fixed function of that right half. These eight nibbles are therefore independent and uniform, even after conditioning on $e'$ for a fixed~$e$. Addition of the round subkey preserves this property. Using the propagation in Fig.~\ref{fig:fig5}, the seven output differences $e_0,e_1,e_2,e_4,e_5,e_6,e_7$ of round~28 have distribution $P_e$, while $e'_3$ has distribution $P_{e'}$; all eight are independent given $e,e'$. Also, $\Pr[e'\mid e,L=7]=P_e(e')$.

Let $v\ne0$ be the common ciphertext difference required by $\mathcal{C}_i$ for $i\in\mathcal{A}$. For fixed $e,e',v\in\{1,\ldots,15\}$, set
\[
\begin{aligned}
a&=P_e(v\oplus e), & b&=P_{e'}(v\oplus e), &
c&=P_e(v\oplus e\oplus e'),\\
d&=P_e(v), & h&=P_e(e).
\end{aligned}
\]
Here, we denote $F_i(e,e',v)$ the conditional probability that a fault $e$ injected into $X_7^{27}$ produces a ciphertext- difference pattern  satisfying $C_i$,  with common nonzero difference $v$, given that the output difference of $S_0^{27}$ is $e'$.
By Equation~\eqref{equ:ciphertext_difference}, we have
\[
\begin{aligned}
F_0(e,e',v)&=a^4bc(1-a),\\
F_2(e,e',v)&=ha^3b(1-a)(1-c),\\
F_3(e,e',v)&=a^4b(1-a)(1-d)(1-c),\\
F_{15}(e,e',v)&=a^5bc.
\end{aligned}
\]
Take $F_0(e,e',v)$ as an example, $\mathcal{C}_0$ requires $\Delta_{C_0},\Delta_{C_2},\Delta_{C_3},\Delta_{C_5}$ to equal $v$, contributing $a^4$. The requirements on $\Delta_{C_1}$ and $\Delta_{C_7}$ contribute $b$ and $c$, while $\Delta_{C_6}\ne v$ contributes $1-a$. The condition $\Delta_{C_4}\ne0$ holds automatically because $e\ne0$ and the S-box is bijective. For $\mathcal{C}_2$, $\Delta_{C_0}=0$ contributes $h$, the four equal differences contribute $a^3b$, and the inequalities contribute $(1-a)(1-c)$. The expressions for $\mathcal{C}_3$ and $\mathcal{C}_{15}$ follow in the same way, including the factor $1-d$ for $\Delta_{C_4}\ne v$ in $\mathcal{C}_3$.

Different values of $v$ define disjoint events. Averaging over $e$ and weighting $e'$ by $P_e(e')$ gives
\[
q_i=\frac1{15}\sum_{e=1}^{15}\sum_{e'=1}^{15}P_e(e')
       \sum_{v=1}^{15}F_i(e,e',v),\qquad i\in\mathcal{A}.
\]
Substitution of the DDT in Table~\ref{tab:ddt} yields the exact values
\[
\begingroup
\setlength{\arraycolsep}{6pt}
\begin{array}{ccccc}
\hline
\rule[-5pt]{0pt}{16pt}
i & 0 & 2 & 3 & 15 \\
\hline
\rule[-12pt]{0pt}{32pt}
q_i &
\dfrac{227}{25165824} &
\dfrac{5201}{62914560} &
\dfrac{91085}{402653184} &
\dfrac{281}{125829120} \\[4pt]
\hline
\end{array}
\endgroup
\]
Proposition~\ref{prop1} gives $q_i=0$ for $i\in\mathcal{J}$. Since the four error events are disjoint and faults at all branches other than~7 are classified correctly,
\[
\Pr[\widehat L\ne L]
=\Pr[L=7]\sum_{i\in\mathcal{A}}q_i
=\frac1{16}\,\frac{644513}{2013265920}
=\frac{644513}{32212254720}.
\]
Taking the complement proves the stated accuracy.
\end{proof}

The probability of a correct assignment \emph{after observing a particular condition} is different from the overall accuracy. Since $\Pr[\mathcal{C}_i\mid L=i]=1$, Bayes' rule under the uniform location prior gives
\[
\Pr[L=i\mid\mathcal{C}_i]=\frac1{1+q_i},\qquad
\Pr[L=7\mid\mathcal{C}_i]=\frac{q_i}{1+q_i},\qquad i\in\mathcal{A}.
\]
For example, $\Pr[L=3\mid\mathcal{C}_3]\approx99.9774\%$. This conditional accuracy should not be confused with the approximately $99.9980\%$ overall accuracy across all sixteen branches. Neither quantity alone specifies the success probability of the subsequent multi-fault key-recovery procedure.

\vspace{5pt}
\noindent\textbf{Statistical results from experiments.}
In the $2^{20}$ random fault simulations over random plaintexts, keys, fault values, and locations, the decision rule based on Table~\ref{tab:fault_branch_conditions} misclassified 18 faults, giving an observed accuracy of $99.9983\%$. Under the probability model of Proposition~\ref{prop2}, the expected number of errors in $2^{20}$ trials is approximately $20.98$. The observed count is consistent with the DDT-weighted calculation.

\section{Fault-Value Determination and Subkey Constraints}
\label{sec:4}
Once the fault branch has been identified from the ciphertext differences according to Propositions~\ref{prop1} and~\ref{prop2}, the corresponding fault propagation path and output ciphertext differences can be further derived. Due to the strong diffusion capability of the linear layer in the EGFN-based structure, the injected fault propagates over three rounds, leading to different correlations between the fault values and the observable ciphertext differences.

Consequently, whether the fault values $e$ and $e'$ can be directly determined from the ciphertext differences depends strongly on the actual fault branch, where $e'$ denotes the output difference of the S-box $S^{27}_{7-i}$ with input difference $e$:
\begin{itemize}
    \item When a fault is injected into branch $X^{27}_j$ for $j \in \{0,\ldots,15\}\setminus\{1,7\}$, both $e$ and $e'$ can be uniquely determined from the ciphertext differences.
    \item When a fault is injected into branch $X^{27}_1$, only $e$ can be uniquely determined, while $e'$ remains uncertain.
    \item When a fault is injected into branch $X^{27}_7$, only a set of candidate pairs for $(e, e')$ can be obtained.
\end{itemize}

In the following subsections, these three scenarios are systematically analyzed to determine the corresponding values of $e$ and $e'$ and the candidate sets for $RK^{29}_i$ and $RK^{28}_i$.

\subsection{Unique Determination of Both $e$ and $e'$}\label{sect:4.1}
When a fault is injected into branch $X^{27}_j$ for $j \in \{0,\ldots,15\}\setminus\{1,7\}$, both $e$ and $e'$ can be uniquely determined based on the ciphertext differences. For clarity, this subsection illustrates the analysis by taking the fault injection into branch $X^{27}_0$ as a representative example.

Fig.~\ref{fig:x0} illustrates the difference propagation when a nibble fault is injected at branch $X^{27}_0$, where $e$, $e'$, and $e'_0$ denote the injected fault, the output difference of S-box $S^{27}_7$ with input difference $e$ and the output difference of S-box $S^{28}_0$ with input difference $e'$, respectively. 
\begin{figure}[!t]
    \centering
    \begin{minipage}{0.49\textwidth}
        \centering
        \includegraphics[width=\linewidth]{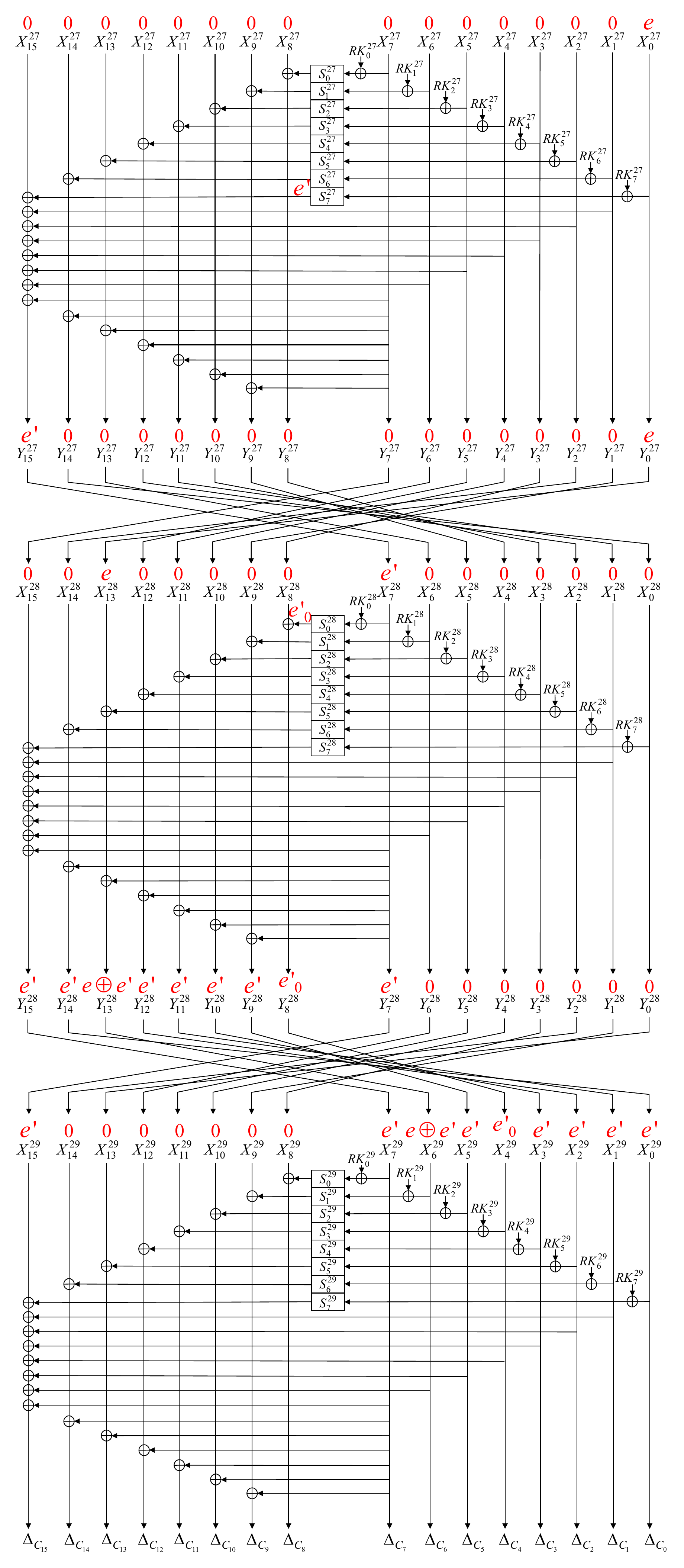}
        \caption{Fault propagation for fault injection at $X^{27}_0$.}
        \label{fig:x0}
    \end{minipage}
    \hfill
    \begin{minipage}{0.49\textwidth}
        \centering
        \includegraphics[width=\linewidth]{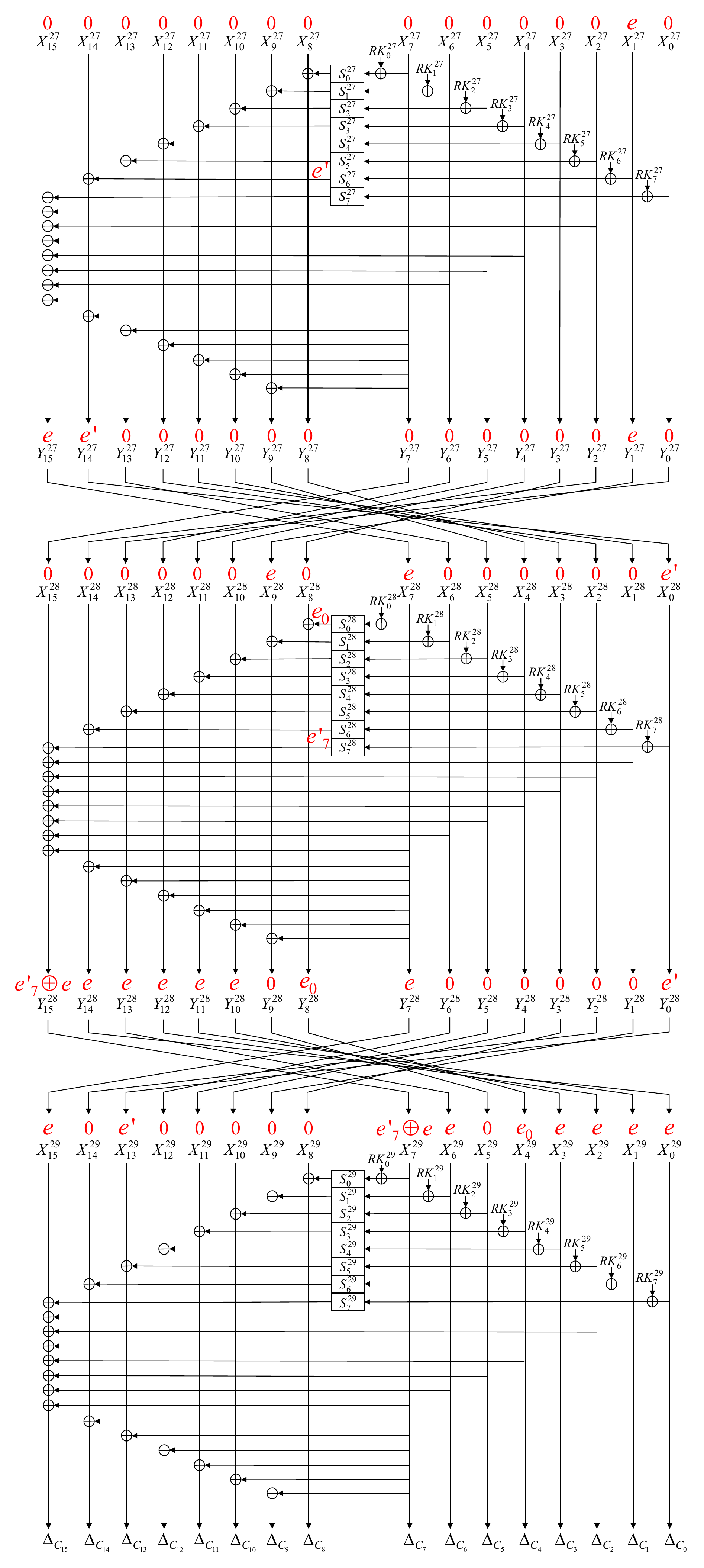}
        \caption{Fault propagation for fault injection at $X^{27}_1$.}
        \label{fig:x1}
    \end{minipage}
\end{figure}
Based on the known ciphertext differences, we can directly obtain the values of $(e,e',e'_0)$ from
\[
\begin{cases}
	\Delta_{C_4} = e'_0,\\
	\Delta_{C_i} = e' \text{ for }i\in\{0,1,2,3,5,7\},\\
	\Delta_{C_6} = e\oplus e',
\end{cases} 
\]
which enables us to construct fault equations for $RK^{28}_0$ and $RK^{29}_i(0\le i\le 7)$. In particular, we obtain $RK^{29}_i \in (C_{7-i} \oplus \mathtt{IN}(\alpha^{29}_{i}, \beta^{29}_{i})),$
where $\alpha^{29}_i = \Delta_{C_{7-i}} \text{ for } 0 \le i \le 7$ and
\[
\beta^{29}_i=
\begin{cases}
	\Delta_{C_8}, & \text{if } i = 0, \\
	\Delta_{C_{15}} \oplus \Delta_{C_4} \oplus \Delta_{C_6}, & \text{if } i = 7,\\
	\Delta_{C_{8+i}} \oplus \Delta_{C_7}, & \text{otherwise}.\\
\end{cases}
\]
Accordingly, $RK^{28}_0$ is constrained by 
\begin{equation*}\label{equ:rk_28_0}
	RK^{28}_0 \in (X^{29}_{15}\oplus \mathtt{IN}(\Delta_{C_7},\Delta_{C_4})).
\end{equation*}
\subsection{Unique Determination of $e$ with Uncertain $e'$}\label{sect:4.2}
When a fault is injected into branch $X^{27}_1$, the corresponding propagation trail shown in Fig.~\ref{fig:x1} indicates that the injected fault can be directly determined from the observed ciphertext differences as $e = \Delta_{C_0}$.

Although the intermediate difference $e'$ cannot be uniquely determined from a single fault session, its candidate set can be derived using the differential propagation constraints and the S-box Differential Distribution Table (DDT). Since $e'$ is the output difference of S-box $S^{27}_6$ with input difference $e$, it must satisfy
\begin{equation}\label{eq:x1_e1}
    \mathtt{DDT}[e][e'] > 0.
\end{equation}
Moreover, the relation $\Delta_{C_7}=e'_7\oplus e$ can be rewritten as $e'_7=\Delta_{C_7}\oplus e$. Since $e'_7$ is the output difference of S-box $S^{28}_7$ with input difference $e'$, the following constraint must hold:
\begin{equation}\label{eq:x1_e2}
    \mathtt{DDT}[e'][\Delta_{C_7}\oplus e] > 0.
\end{equation}
Furthermore, from the internal differential relations along the propagation trail, the output difference of $S^{29}_5$ is constrained by $S^{29}_5(\Delta_{C_2}) = e' \oplus \Delta_{C_7} \oplus \Delta_{C_{13}}$, which yields the implicit DDT constraint
\begin{equation}\label{eq:x1_e4}
    \mathtt{DDT}[\Delta_{C_2}][e' \oplus \Delta_{C_7} \oplus \Delta_{C_{13}}] > 0.
\end{equation}
Since $e'$ is a nonzero 4-bit value, it ranges from $1$ to $15$. By traversing all $15$ hypotheses and retaining only those that simultaneously satisfy the DDT constraints from (\ref{eq:x1_e1}) to (\ref{eq:x1_e4}), a refined candidate set for $e'$, denoted by $\mathcal{E}_{e'}$, is obtained.

For each candidate value in $\mathcal{E}_{e'}$, the input differences $\alpha$ and output differences $\beta$ of the S-boxes in the final two rounds can be derived. For round $29$, the input differences are given by $\alpha^{29}_i=\Delta_{C_{7-i}}$ for $0\leq i\leq 7$, and the corresponding output differences are
\begin{equation*}
\beta^{29}_i =
\begin{cases}
\Delta_{C_8}, & \text{if }i = 0, \\
\Delta_{C_{13}}\oplus \Delta_{C_7} \oplus e', & \text{if }i = 5, \\
\Delta_{C_{15}} \oplus \bigoplus_{j=1}^{7}\Delta_{C_j} \oplus e, & \text{if } i=7,\\
\Delta_{C_{8+i}}\oplus \Delta_{C_7}, & \text{otherwise}.
\end{cases}
\end{equation*}
Since the true value of $e'$ is contained in $\mathcal{E}_{e'}$, the overall candidate set for $RK^{29}$ under this single fault session is obtained by taking the union of the key candidates derived from all possible choices of $e'$:
\begin{equation*}
RK^{29}_i \in \bigcup_{e' \in \mathcal{E}_{e'}} 
\left\{ C_{7-i} \oplus \mathtt{IN}(\alpha^{29}_{i}, \beta^{29}_{i}) \right\}.
\end{equation*}
For round $28$, only the S-boxes $S^{28}_0$ and $S^{28}_7$ are activated along the propagation trail. Therefore, by taking the union over $\mathcal{E}_{e'}$, the local candidate spaces for the corresponding subkey nibbles of $RK^{28}$ are bounded by
\begin{align*}
RK^{28}_0 \in X^{28}_{7} \oplus \mathtt{IN}(\Delta_{C_0}, \Delta_{C_4}), \qquad
RK^{28}_7 \in \bigcup_{e' \in \mathcal{E}_{e'}} \left\{ X^{28}_{0} \oplus \mathtt{IN}(e', \Delta_{C_7} \oplus \Delta_{C_0}) \right\}.
\end{align*}
\subsection{Candidate Set Estimation for Both $e$ and $e'$}\label{sect:4.3}
When a fault is injected into branch $X^{27}_7$, the ciphertext differences are insufficient for unique identification, allowing only the determination of a candidate pair set for $(e, e')$.

Fig.~\ref{fig:fig5} depicts the difference propagation when a nibble fault is injected at branch $X^{27}_7$, where $e \in \mathbb{F}_2^4$ denotes the non-zero fault and $e'\in \mathbb{F}_2^4$ denotes the corresponding output difference of $S^{27}_0$. 
\begin{figure*}[!h]
	\centering
	\includegraphics[width=0.7\textwidth]{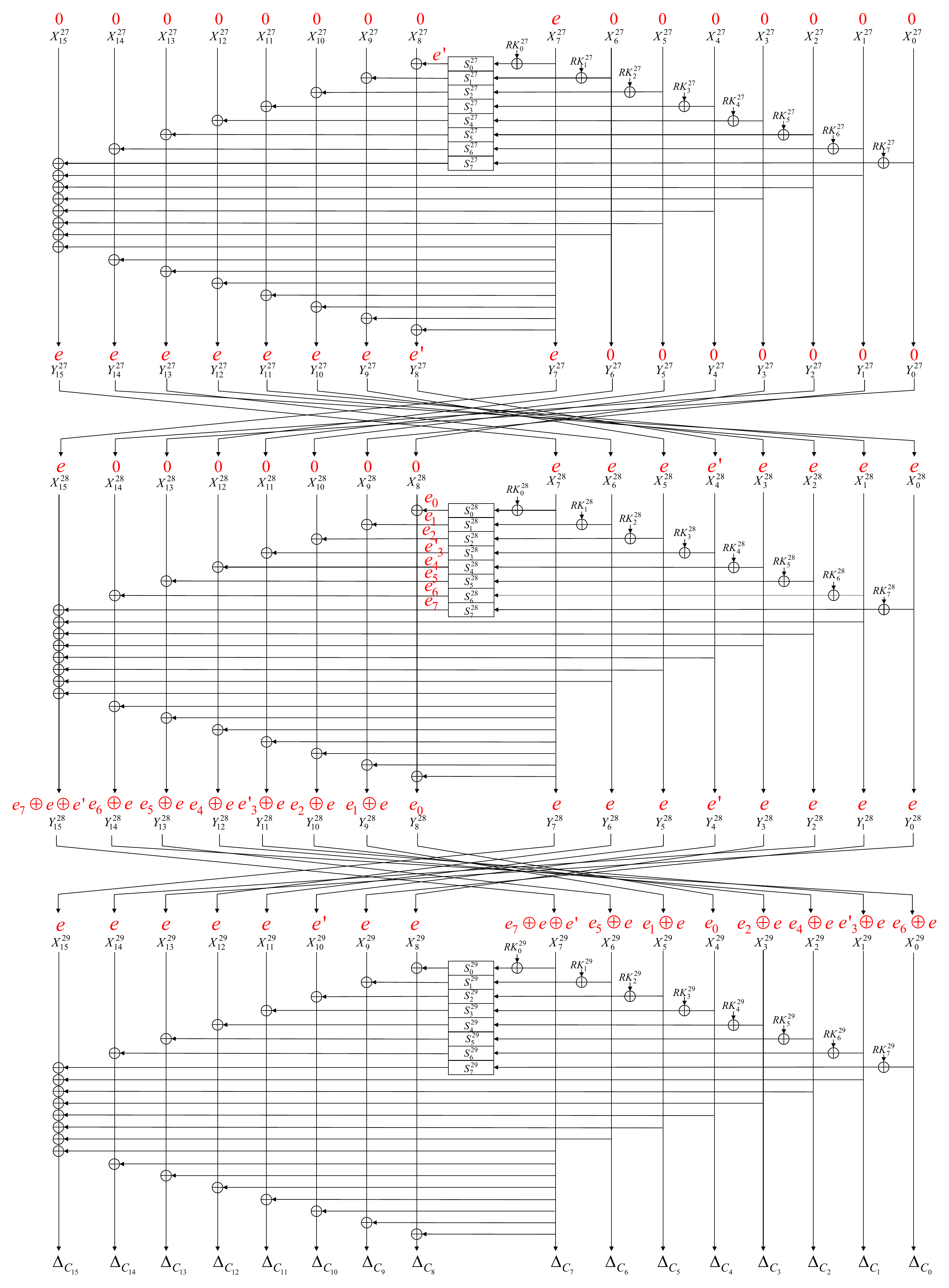}
	\caption{Fault propagation after fault injection at $X^{27}_7$.}
	\label{fig:fig5}
\end{figure*}
As shown in Fig.~\ref{fig:fig5}, we have
\begin{equation*}
	e' = S(X^{27}_7 \oplus e \oplus RK^{27}_{0}) \oplus S(X^{27}_7 \oplus RK^{27}_{0}).
\end{equation*}
According to the S-box's DDT, as shown in Table~\ref{tab:ddt}, the condition $\mathtt{DDT}[e][e'] > 0$ must hold for $e$ and $e'$.
After the linear transformation and permutation layer, the input differences in round 28 are given by
\begin{equation*}
	\Delta_{X^{28}_i} =
	\begin{cases}
		e', & \text{if }i = 4, \\
		0, & \text{if }8 \le i \le 14, \\
		e, & \text{otherwise}.
	\end{cases}
\end{equation*}
Consequently, seven S-boxes in round 28 have input difference $e$, whereas $S_{3}^{28}$ has the input difference $e'$. Let $e_i$ denote the output difference of $S_{i}^{28}$ for $i\in\{0, 1, 2, 4, 5, 6, 7\}$, and $e'_3$ denote the output difference of $S_{3}^{28}$. Then, after the linear layer and permutation layer, the input differences in round 29 are \begin{equation*}
	\Delta_{X_{i}^{29}}=
	\begin{cases}
		e_{\pi^{-1}(i)-8}\oplus e, & \text{if }i\in\{0,2,3,5,6\},\\
		e'_3\oplus e, & \text{if }i=1,\\
		e_0, & \text{if }i=4,\\
		e_7\oplus e\oplus e', & \text{if }i=7,\\
		e', & \text{if }i=10,\\
		e, &  \text{otherwise},
	\end{cases}
\end{equation*}
where $\pi^{-1}$ denotes the inverse of the permutation layer.

Subsequently, after the nonlinear and linear layers, the ciphertext differences can be represented as 
\begin{equation}\label{equ:ciphertext_difference}
\begin{aligned}
	&\Delta_{C_{0}} = e_6 \oplus e , &\qquad&  \Delta_{C_{8}} = S^{29}_0(\Delta_{C_{7}}) \oplus e, \\
	&\Delta_{C_{1}} = e'_3 \oplus e , && \Delta_{C_{9}} = S^{29}_1(\Delta_{C_{6}}) \oplus e \oplus \Delta_{C_{7}}, \\
	&\Delta_{C_{2}} = e_4 \oplus e , && \Delta_{C_{10}} = S^{29}_2(\Delta_{C_{5}}) \oplus e' \oplus \Delta_{C_{7}}, \\
	&\Delta_{C_{3}} = e_2 \oplus e , &&  \Delta_{C_{11}} = S^{29}_3(\Delta_{C_{4}}) \oplus e \oplus \Delta_{C_{7}}, \\
	&\Delta_{C_{4}} = e_0 , &&  \Delta_{C_{12}} = S^{29}_4(\Delta_{C_{3}}) \oplus e \oplus \Delta_{C_{7}}, \\
	&\Delta_{C_{5}} = e_1 \oplus e , &&  \Delta_{C_{13}} = S^{29}_5(\Delta_{C_{2}}) \oplus e \oplus \Delta_{C_{7}},\\
	&\Delta_{C_{6}} = e_5 \oplus e , &&  \Delta_{C_{14}} = S^{29}_6(\Delta_{C_{1}}) \oplus e \oplus \Delta_{C_{7}},\\
	&\Delta_{C_{7}} = e_7 \oplus e \oplus e' , &&  \Delta_{C_{15}} = S^{29}_7(\Delta_{ C_0}) \oplus \bigoplus_{i=1}^{7} \Delta_{C_i} \oplus e,
\end{aligned}
\end{equation}
where $S^{29}_i(\Delta_{C_{7-i}})$ for $0\le i\le 7$ denotes the output difference of the S-box $S^{29}_i$ when the input difference is $\Delta_{C_{7-i}}$.
Different from Section~\ref{sect:4.1}, we cannot directly deduce the value of $e$ because $e'$, $e_i$ and the output differences of the S-boxes are all unknown. Next, we will detail how to exploit the above relation as well as DDT of the S-box to determine the value of $e$.

\vspace{5pt}
\noindent\textbf{Candidate values of $e$}.
It is worth noting that $e_i (0\le i\le 7)$ denotes the output difference of the S-box $S^{28}_i$ when the input difference is $e$. Taking $\Delta_{C_0} = e_6 \oplus e$ as an example, we can re-write this equation as $e_6 = \Delta_{C_0} \oplus e$. Combining the DDT implies that the condition $\mathtt{DDT}[e][e_6] > 0$ must hold, therefore, $\mathtt{DDT}[e][\Delta_{C_{0}}\oplus e]> 0$. As a result, we can derive a set of DDT-based constraints on $e$ from the equations for $\Delta_{C_0},\Delta_{C_2},\cdots,\Delta_{C_6}$ as follows:
\begin{equation}\label{equ:constrain_e1}
	\left\{
	\begin{aligned}
		&\mathtt{DDT}[e][\Delta_{C_0} \oplus e] > 0, \\
		&\mathtt{DDT}[e][\Delta_{C_2} \oplus e] > 0,\\
		&\mathtt{DDT}[e][\Delta_{C_3} \oplus e] > 0, \\
		&\mathtt{DDT}[e][\Delta_{C_4}] > 0, \\
		&\mathtt{DDT}[e][\Delta_{C_5} \oplus e] > 0, \\
		&\mathtt{DDT}[e][\Delta_{C_6} \oplus e] > 0.
	\end{aligned}
	\right.
\end{equation}
Since $e$ is a 4-bit non-zero value, its possible values range from 1 to 15 (15 candidates in total). By traversing all possible values of $e$ and retaining only those that simultaneously satisfy the inequalities in Equation~\eqref{equ:constrain_e1}, a set of candidate values for $e$ is obtained.

To further narrow down the candidate set, the equation involving $\Delta_{C_i}$ for $8\le i\le 15$ can be used for additional filtering. For instance, let us focus on $\Delta_{C_8}$. If we can identify a set of candidates for the value of $S^{29}_0(\Delta_{C_7})$, then we can use this set as a new filter for $e$ due to the equation $\Delta_{C_{8}} = S^{29}_0(\Delta_{C_{7}}) \oplus e$. Combining the relation for $\Delta_{C_8}$ with those for $\Delta_{C_9}$, $\Delta_{C_{11}}$, $\Delta_{C_{12}}$, $\Delta_{C_{13}}$, $\Delta_{C_{14}}$ and $\Delta_{C_{15}}$, we obtain a new system of equations as
\begin{equation*}\label{eq:2}
	\left\{
	\begin{aligned}
		&x = S^{29}_0(\Delta_{C_7}),\\
		&\Delta_{C_{8}} \oplus \Delta_{C_{9}} = x \oplus S^{29}_{1}(\Delta_{C_{6}}) \oplus \Delta_{ C_{7}} ,\\
		&\Delta_{C_{8}} \oplus \Delta_{ C_{11}} = x \oplus S^{29}_{3}(\Delta_{C_{4}}) \oplus \Delta_{C_{7}},\\
		&\Delta_{C_{8}} \oplus \Delta_{ C_{12}} = x \oplus S^{29}_{4}(\Delta_{C_{3}}) \oplus \Delta_{C_{7}} ,\\
		&\Delta_{C_{8}} \oplus \Delta_{ C_{13}} = x \oplus S^{29}_{5}(\Delta_{C_{2}}) \oplus \Delta_{C_{7}} ,\\
		&\Delta_{C_{8}} \oplus \Delta_{ C_{14}} = x \oplus S^{29}_6(\Delta_{C_1}) \oplus \Delta_{C_{7}} ,\\
		&\Delta_{C_{8}} \oplus \Delta_{ C_{15}} = x \oplus S^{29}_{7}(\Delta_{C_{0}}) \oplus \bigoplus_{i=1}^{7} \Delta_{C_i}.
	\end{aligned}
	\right.
\end{equation*}
Since all ciphertext differences $\Delta_{C_i}$ are known, the possible values for $x$ satisfying the above system of equations can be identified based on the S-box's DDT, resulting in an additional candidate set for $e$. Taking the intersection of this candidate set 
with the one previously derived from Equation~\eqref{equ:constrain_e1} yields a more precise candidate set for $e$.

\vspace{5pt}
\noindent\textbf{Candidate values of $e'$.}
After determining the candidate values for $e$, we next derive the candidate values for $e'$. First, since $e'$ is the output difference corresponding to the input difference $e$ of the S-box $S^{27}_0$, the basic DDT constraint
\begin{equation}
	\mathtt{DDT}[e][e'] > 0
	\label{eq:12}
\end{equation}
must hold.
Moreover, from the relation $\Delta_{C_1} = e'_3 \oplus e$, we obtain 
\begin{equation}
	\mathtt{DDT}[e'][\Delta_{C_1} \oplus e] > 0
	\label{eq:13}
\end{equation}
since $e'_3$ is the output difference of the S-box $S^{28}_3$ with the input difference $e'$.
Furthermore, from the relation $\Delta_{C_7} = e_7 \oplus e' \oplus e$, we have 
\begin{equation}
	\mathtt{DDT}[e][\Delta_{C_7} \oplus e' \oplus e] > 0
	\label{eq:14}
\end{equation}
as $e_7$ is the output difference of the S-box $S^{28}_7$ with the input difference $e$.
Similarly, the relation $\Delta_{C_{10}} = S^{29}_2(\Delta_{C_5}) \oplus e' \oplus \Delta_{C_7}$ gives the constraint
\begin{equation}
	\mathtt{DDT}[\Delta_{C_5}][e' \oplus \Delta_{C_{10}} \oplus \Delta_{C_7}] > 0.
	\label{eq:15}
\end{equation}
By combining the constraints from~\eqref{eq:12} to~\eqref{eq:15} with the already determined candidate set for $e$, a filtered candidate set for $e'$ can be obtained. 
\subsubsection{Simultaneous Subkeys Recovery from Candidate Pairs of $(e,e')$}
When a fault $e$ is injected at $X^{27}_7$, the candidate sets for the fault $e$ and the intermediate difference $e'$ can be obtained as described in the last subsubsection. Consequently, a refined candidate set for $(e,e')$, denoted by $\mathcal{E}_{(e,e')}$, is obtained. We can enumerate all possible pairs $(e,e')$ and sequentially derive candidate values for the subkeys $RK^{29}$ and $RK^{28}$. For each pair, we take the union of the corresponding candidates for $RK^{29}$ and $RK^{28}$. 

\vspace{5pt}
\noindent\textbf{Recovering $RK^{29}$.}
From the round function, we can directly get the input differences of all S-boxes in round 29 as $\alpha^{29}_i = \Delta_{C_{7-i}}$ for $0\le i \le 7$. As for the output differences, for each candidate pair $(e,e')\in\mathcal{E}_{(e,e')}$, we have  
\[
\beta^{29}_i(e,e') = 
\begin{cases}
	\Delta_{C_8}\oplus e, & \text{if } i=0, \\
	\Delta_{C_{10}}\oplus \Delta_{C_7} \oplus e', & \text{if } i=2,\\
	\Delta_{C_{15}} \oplus \bigoplus_{j=1}^{7}\Delta_{C_j} \oplus e, & \text{if } i=7,\\
	\Delta_{C_{8+i}}\oplus \Delta_{C_7} \oplus e, & \text{otherwise.}
\end{cases}
\]
In addition, the input and output differences for each S-box can also be  calculated by
\[
\beta^{29}_i = S(X^{29}_{7-i}\oplus RK^{29}_i) \oplus S(X^{29}_{7-i}\oplus RK^{29}_i\oplus \alpha^{29}_i).
\]
Therefore, we have the relation 
\begin{equation*}
	RK^{29}_i \oplus X^{29}_{7-i} \in \mathtt{IN}(\alpha^{29}_{i},\beta^{29}_{i}(e,e'))
\end{equation*}
for the subkey $RK^{29}_i$. Noting that $X^{29}_{7-i}$ is equal to $C_{7-i}$, we finally derive the fault equation
\begin{equation*}
	RK^{29}_i \in (C_{7-i} \oplus \mathtt{IN}(\alpha^{29}_{i},\beta^{29}_{i}(e,e'))).
\end{equation*}
Since the true fault value $e$ is unknown but must belong to the candidate set of $e$, the correct key must appear in at least one candidate set derived from the pairs $(e,e')$. Therefore, for a single fault injection, the overall candidate key set is obtained by taking the union of the key candidate sets corresponding to all possible $(e,e')$ pairs:
\begin{equation}
RK^{29}_i
\in
\bigcup_{(e,e')\in\mathcal{E}_{(e,e')}}
\left\{
	C_{7-i}\oplus {\mathtt{IN}}\bigl(\alpha^{29}_i,\beta^{29}_i(e,e')\bigr)
	\right\}.
\end{equation} 

\noindent\textbf{Recovering $RK^{28}$.}
After $RK^{29}$ has been recovered, we can then determine the correct pair $(e,e')$, and proceed to decrypt the last round to generate the intermediate state difference $(\Delta_{Y^{28}_{15}}, \Delta_{Y^{28}_{14}}, \ldots, \Delta_{Y^{28}_0})$. On the one hand, the input and output differences of all S-boxes in the 28-th round can be represented as $\alpha^{28}_i = \Delta_{Y^{28}_{7-i}}$ for $0 \le i \le 7$ and 
\[
\beta^{28}_i = 
\begin{cases}
	\Delta_{Y^{28}_8}, & \text{if } i=0, \\
	\Delta_{Y^{28}_{15}} \oplus \Delta_{Y^{28}_{7}} \oplus \Delta_{Y^{28}_{4}},&  \text{if } i=7,\\
	\Delta_{ Y^{28}_{8+i}} \oplus \Delta_{ Y^{28}_7}, & \text{otherwise,}
\end{cases}
\]
respectively. On the other hand, the output difference for each S-box is computed by
\[
\beta^{28}_i = S(X^{28}_{7-i}\oplus RK^{28}_i) \oplus S(X^{28}_{7-i}\oplus RK^{28}_i\oplus \alpha^{28}_i).
\]
Therefore, we have the relation 
\begin{equation*}
	RK^{28}_i \oplus X^{28}_{7-i} \in \mathtt{IN}(\alpha^{28}_{i},\beta^{28}_{i})
\end{equation*}
for the subkey $RK^{28}_i$, finally derive the fault equation
\begin{equation*}
	RK^{28}_i \in (Y^{28}_{7-i} \oplus \mathtt{IN}(\alpha^{28}_{i},\beta^{28}_{i})).
\end{equation*}

\subsection{Multi-Fault Intersection Strategy and Key Recovery}
During the key-recovery process, a global candidate set is maintained for each subkey nibble $RK^{29}_i$ for $0 \leq i \leq 7$, and each set is initialized as the full set of $16$ possible nibble values. After each fault injection, the candidate set of $RK^{29}_i$ is updated through the following steps:
\begin{enumerate}
\item Identify the fault injection branch from the ciphertext difference;
\item Derive the constraint equations according to the corresponding fault propagation path and compute the local feasible set;
\item Update the global candidate set by intersecting it with the local feasible set whenever the local feasible set is nonempty.
\end{enumerate}
This process is repeated until each subkey nibble of $RK^{29}$ is uniquely determined, thereby recovering the complete subkey $RK^{29}$.

After $RK^{29}$ has been fully determined, the final round is partially decrypted to recover the intermediate state $Y^{28}$ and its corresponding differences $\Delta_{Y^{28}}$. The same faulty ciphertext used for recovering $RK^{29}$ is then reused for recovering $RK^{28}$. Specifically, the global candidate set for each $RK^{28}_i$ is initialized as the full nibble space, and the local candidate sets are constructed according to the identified fault-injection branch. Similarly, each global candidate set is updated by intersecting it with the corresponding local candidate set whenever the latter is nonempty. Through this multi-fault intersection process, the candidate space for each subkey nibble is progressively reduced until it converges to a unique value, leading to the full recovery of $RK^{28}$.

Finally, given the recovered subkeys $RK^{29}$ and $RK^{28}$, $60$ bits of the 80-bit master key can be straightforwardly derived according to the key schedule. The remaining $20$ bits are then resolved by exhaustive search over $2^{20}$ candidates, thus fully recovering the master key $MK$.

\section{Experimental Evaluation and Bottleneck Analysis}\label{sec:5}
To evaluate the number of faults required for subkey recovery under random fault injection over all branches in the 27-th round, we conducted $2^{15}$ independent simulation-based experiments. In each experiment, faults were randomly injected into branches $X^{27}_0$ to $X^{27}_{15}$, and the injected branch was inferred from the resulting ciphertext difference before key recovery. For each successful experiment, i.e., an experiment in which both $RK^{29}$ and $RK^{28}$ were correctly recovered, we recorded the number of faults required for recovery and computed its proportion with respect to all $2^{15}$ experiments.

The results show that 9 out of the $2^{15}$ experiments failed to recover the subkeys even when the number of injected faults was increased to 1000. Further analysis indicates that these failures were caused by incorrect fault-branch identification. Specifically, an incorrect branch decision introduces candidate constraints that are inconsistent with the actual propagation path, which may eliminate the true subkey value during the intersection of candidate sets and consequently lead to recovery failure. This observation is consistent with the success probability of the proposed ciphertext-difference-based fault-branch identification method. For the experiments in which the subkeys were successfully recovered, the distribution of the required number of faults is shown in Fig.~\ref{fig:fig12}(a).

\begin{figure}[!t]
	\centering
	\begin{minipage}[t]{0.485\linewidth}
		\centering
		\includegraphics[width=0.93\linewidth]{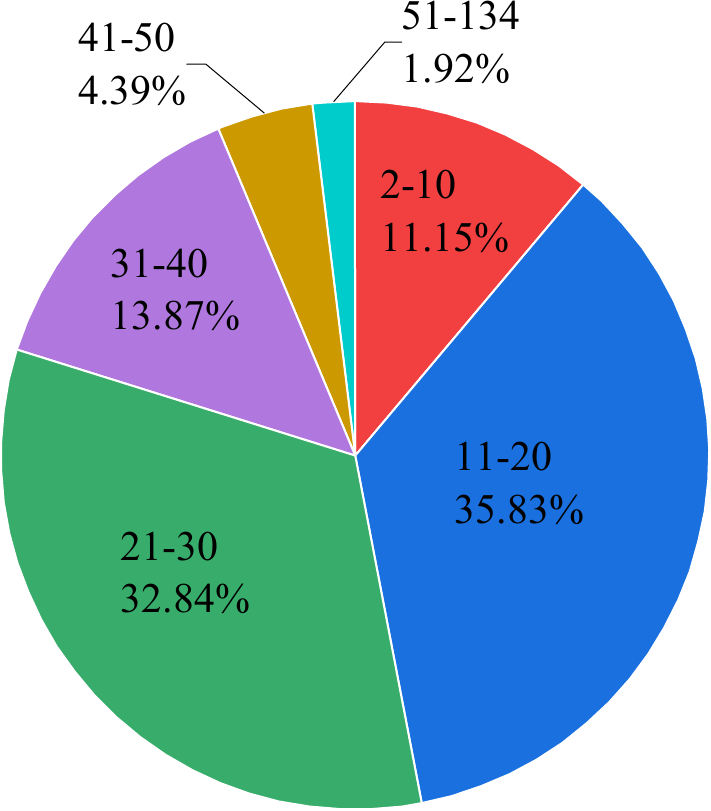}\\[-1mm]
		\small (a) Recovering $RK^{29}$ and $RK^{28}$.
	\end{minipage}\hfill
	\begin{minipage}[t]{0.485\linewidth}
		\centering
		\includegraphics[width=0.93\linewidth]{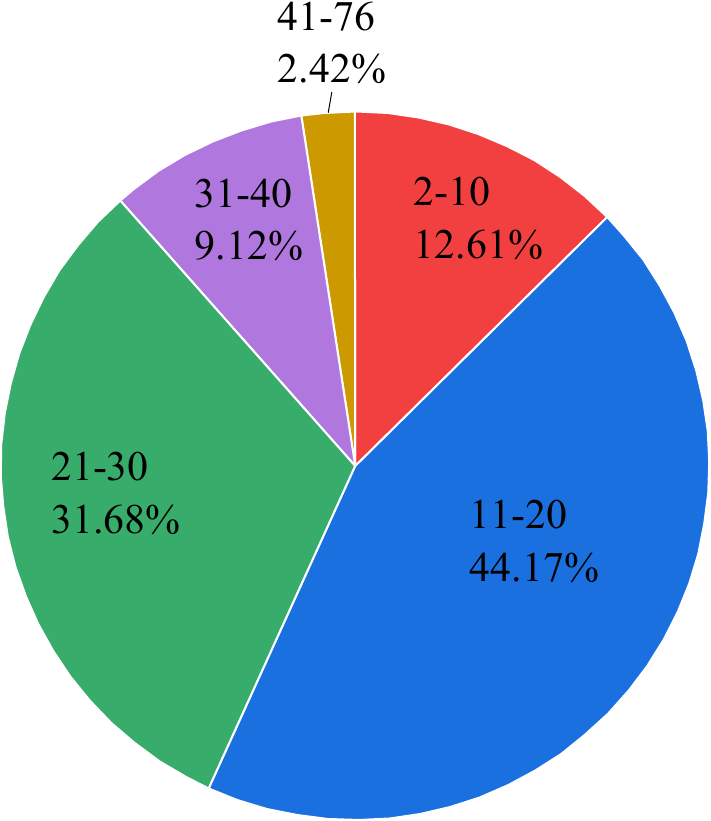}\\[-1mm]
		\small (b) Recovering $RK^{29}$ and $RK^{28}$ excluding $RK^{28}_3$.
	\end{minipage}
	\caption{Distribution of fault counts in the two key-recovery settings. For exact statistical values, please refer to Tables~\ref{tab:distribution31} and~\ref{tab:distribution32} in Part~\ref{sm:1} of Supplementary Material.}
	\label{fig:fig12}
\end{figure}

As discussed in Section~\ref{sec:4}, whether the fault values $e$ and $e'$ can be uniquely determined from the ciphertext differences depends on the actual injected branch. In particular, the recovery of $RK^{28}_3$ is more difficult than that of the other subkey nibbles. This is because the input difference of S-box $S^{28}_3$ depends on the intermediate difference $e'$, whose candidate set is further affected by the candidate value of $e$. As a result, recovering $RK^{28}_3$ requires a noticeably larger number of faults.

To investigate this phenomenon, we perform random fault injection over the full branch range $X^{27}_0$--$X^{27}_{15}$, with a total of $2^{15}$ independent experiments. In each experiment, the injected branch is inferred from the resulting ciphertext difference before key recovery is performed. For each successful experiment, i.e., an experiment in which $RK^{29}$ and $RK^{28}$ excluding the nibble $RK^{28}_3$ are correctly recovered, we record the required number of faults and compute its proportion with respect to all $2^{15}$ experiments. Similar to the previous experiment, 8 experiments still failed to recover the target subkeys even when the number of injected faults was increased to 1000, due to incorrect fault-branch identification. For the successful experiments, the distribution of the required number of faults is shown in Fig.~\ref{fig:fig12}(b).

By comparing Fig.~\ref{fig:fig12}(a) and Fig.~\ref{fig:fig12}(b), it can be observed that excluding $RK^{28}_3$ from the recovery target significantly reduces the required number of faults. To further improve the overall efficiency of the fault injection attack, this paper adopts an exhaustive search to recover $RK^{28}_3$, thereby effectively reducing the actual number of injected faults.
When $RK^{28}_3$ is not recovered by DFA, we enumerate its 16 possible values and combine each guess with the remaining key-schedule inversion. Consequently, the final exhaustive-search cost increases by at most a factor of $2^4$, from $2^{20}$ to $2^{24}$ candidate master keys.
\begin{figure}[!t]
	\centering
	\includegraphics[width=0.86\linewidth]{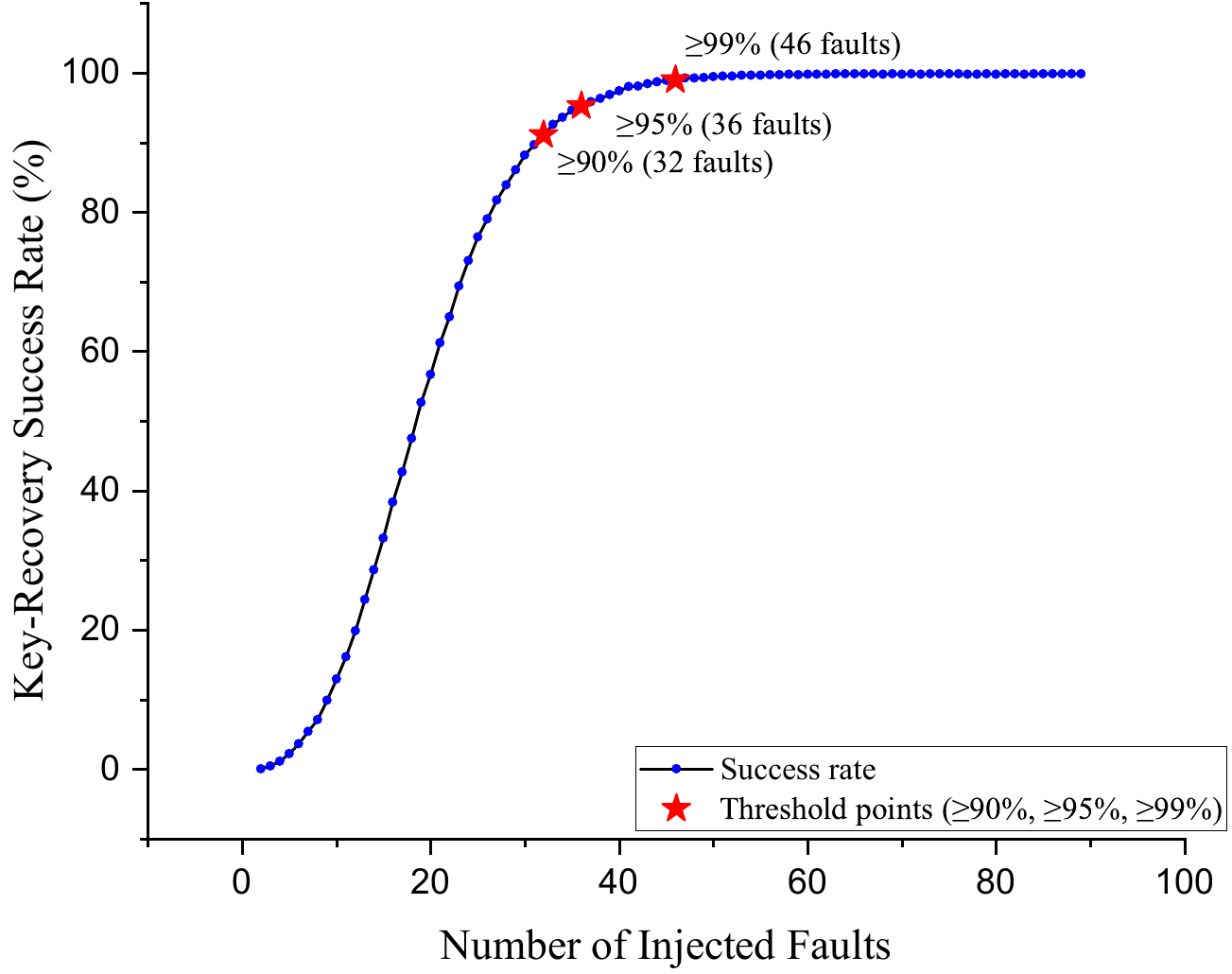}
	\caption{Key-recovery success rates for different numbers of injected faults. For exact statistical values, please refer to Table~\ref{tab:successs_rate3} in Part~\ref{sm:1} of Supplementary Material.}
	\label{fig:fig14}
\end{figure}
Furthermore, to evaluate the key-recovery success rate under a fixed number of injected faults, we conduct another set of $2^{15}$ experiments for each selected fault count. In these experiments, faults are randomly injected over the same full branch range $X^{27}_0$--$X^{27}_{15}$, and the recovery target is $RK^{29}$ together with $RK^{28}$ excluding $RK^{28}_3$. For each fixed number of injected faults, we measure the proportion of experiments in which the target subkeys are successfully recovered. The results are summarized in Fig.~\ref{fig:fig14},

which indicate that the key-recovery success rate reaches 95\% with 36 fault injections and exceeds
99\% when the number of injections is increased to 46. Overall, the success rate increases markedly with the number
of injected faults. All experiments are conducted under the proposed fixed-timing full-branch random-location nibble fault model. The experimental results provide an empirical basis for evaluating the security of \cipher under the considered fault model.

As a special case of the random-location fault model, we also evaluate the key-recovery performance when faults are consistently injected into branch $X_7^{27}$. The corresponding results, including the distribution of the required number of injected faults and the key-recovery success rates for different fault counts, are presented in Part~\ref{sm:fixed_x7} of Supplementary Material.

\section{Conclusion}\label{sec:6}
In this paper, we evaluated the fault resilience of \cipher under a fixed-timing full-branch random-location nibble fault model, where the injection timing is fixed to round~27 but the affected branch is randomly selected from all sixteen state branches and remains unknown to the attacker. The analysis shows that, even without branch-level spatial control, ciphertext differences still preserve sufficient structural information to identify the fault location with high probability and to derive effective constraints on consecutive round subkeys.

Simulation results confirm that \cipher is vulnerable to DFA under this unknown-location fault setting. More importantly, the results reveal that the EGFN structure may exhibit exploitable branch-dependent leakage when faults propagate through multiple branches. This suggests that evaluating only fixed-location faults may underestimate the security implications of spatially imprecise fault injections for multi-branch lightweight ciphers. Our work therefore provides a more conservative algorithmic assessment of \cipher and offers a useful reference for analyzing and protecting other GFN/EGFN-based designs. The study remains simulation-based under an explicit nibble-level fault abstraction, and future work includes validating such models in concrete implementations and developing countermeasures against branch-dependent fault leakage.

%
%
%
\bibliographystyle{splncs04}
\bibliography{mybibliography}
\clearpage

\section*{Supplementary Material}
\addcontentsline{toc}{section}{Supplementary Material}
This part serves as the Supplementary Material for providing details of \cipher's key schedule, experimental results and fault propagation paths.
Specifically,
\begin{itemize}
    \item Part~\ref{sm:sbox_ddt} shows the DDT of \cipher's Sbox.
    \item Part~\ref{sm:key_schedule} gives the details of \textsf{ExtractRoundKey} and \textsf{RoundFnLFSM} functions of \cipher's key schedule.
    \item Part~\ref{sm:1} provides the raw tables underlying the main-text figures, reporting (i) the distribution of the required number of fault injections conditioned on successful recovery of the two target subkeys  (e.g., $RK^{29}$ and $RK^{28}$) and (ii) key-recovery success rates for different fixed numbers of injections. 
    \item Part~\ref{sm:fixed_x7} reports the key-recovery performance when faults are consistently injected into branch $X_7^{27}$, including the distribution of the required number of faults and the success rates for different fault counts.
    \item Part~\ref{sm:2} provides fault-propagation paths for single-branch faults injected at $X^{27}_{i}(i\in\{1,2,...,6,8,...,15\})$, tracing differences through the S-box and diffusion layers to validate the key-recovery propagation relations.
\end{itemize}

\appendix
\section{DDT of \cipher's Sbox}\label{sm:sbox_ddt}
\begin{table}[h]
\centering
\caption{Differential Distribution Table (DDT) of \cipher's S-box.}\label{tab:ddt}%
\begin{tabular}{c|cccccccccccccccc}
\toprule
\diagbox[height=1.5em]{$\Delta_{in}$}{$\Delta_{out}$} & \textbf{0} & \textbf{1} & \textbf{2} & \textbf{3} & \textbf{4} & \textbf{5} & \textbf{6} & \textbf{7} & \textbf{8} & \textbf{9} & \textbf{A} & \textbf{B} & \textbf{C} & \textbf{D} & \textbf{E} & \textbf{F}\\
\midrule
\textbf{0} & 16 & 0 & 0 & 0 & 0 & 0 & 0 & 0 & 0 & 0 & 0 & 0 & 0 & 0 & 0 & 0\\
\textbf{1} & 0 & 2 & 0 & 0 & 0 & 0 & 2 & 0 & 0 & 2 & 2 & 2 & 4 & 0 & 0 & 2 \\
\textbf{2} & 0 & 0 & 0 & 2 & 2 & 0 & 2 & 2 & 0 & 4 & 0 & 2 & 0 & 2 & 0 & 0\\
\textbf{3} & 0 & 2 & 0 & 0 & 0 & 2 & 2 & 2 & 2 & 0 & 0 & 0 & 0 & 2 & 0 & 4 \\
\textbf{4} & 0 & 0 & 0 & 2 & 0 & 2 & 0 & 0 & 0 & 0 & 2 & 4 & 0 & 2 & 2 & 2\\ 
\textbf{5} & 0 & 4 & 2 & 2 & 0 & 2 & 0 & 2 & 0 & 2 & 2 & 0 & 0 & 0 & 0 & 0\\ 
\textbf{6} & 0 & 0 & 2 & 0 & 0 & 0 & 4 & 2 & 0 & 0 & 2 & 0 & 2 & 2 & 2 & 0\\ 
\textbf{7} & 0 & 0 & 0 & 2 & 2 & 2 & 2 & 0 & 2 & 0 & 4 & 0 & 2 & 0 & 0 & 0 \\ 
\textbf{8} & 0 & 2 & 2 & 4 & 2 & 0 & 2 & 0 & 0 & 0 & 0 & 0 & 0 & 0 & 2 & 2\\ 
\textbf{9} & 0 & 0 & 0 & 2 & 0 & 0 & 0 & 2 & 4 & 2 & 0 & 0 & 2 & 0 & 2 & 2 \\ 
\textbf{A} & 0 & 0 & 2 & 0 & 2 & 0 & 0 & 4 & 2 & 0 & 2 & 2 & 0 & 0 & 0 & 2\\ 
\textbf{B} & 0 & 2 & 0 & 0 & 2 & 2 & 0 & 2 & 0 & 0 & 0 & 2 & 2 & 0 & 4 & 0\\ 
\textbf{C} & 0 & 2 & 0 & 0 & 2 & 0 & 0 & 0 & 2 & 2 & 2 & 0 & 0 & 4 & 2 & 0\\ 
\textbf{D} & 0 & 2 & 4 & 2 & 0 & 0 & 0 & 0 & 2 & 0 & 0 & 2 & 2 & 2 & 0 & 0\\ 
\textbf{E} & 0 & 0 & 2 & 0 & 4 & 2 & 0 & 0 & 0 & 2 & 0 & 0 & 2 & 2 & 0 & 2\\ 
\textbf{F} & 0 & 0 & 2 & 0 & 0 & 4 & 2 & 0 & 2 & 2 & 0 & 2 & 0 & 0 & 2 & 0\\
\bottomrule
\end{tabular}
\end{table}
\section{\textsf{ExtractRoundKey} and \textsf{RoundFnLFSM} functions}\label{sm:key_schedule}
The overall structure of the key schedule and the update function \textsf{RoundFnLFSM} are depicted in Fig.~\ref{fig_3} and~\ref{fig_4}, respectively.
\begin{figure}[!htp]
	\centering
	\includegraphics[width=0.7\linewidth]{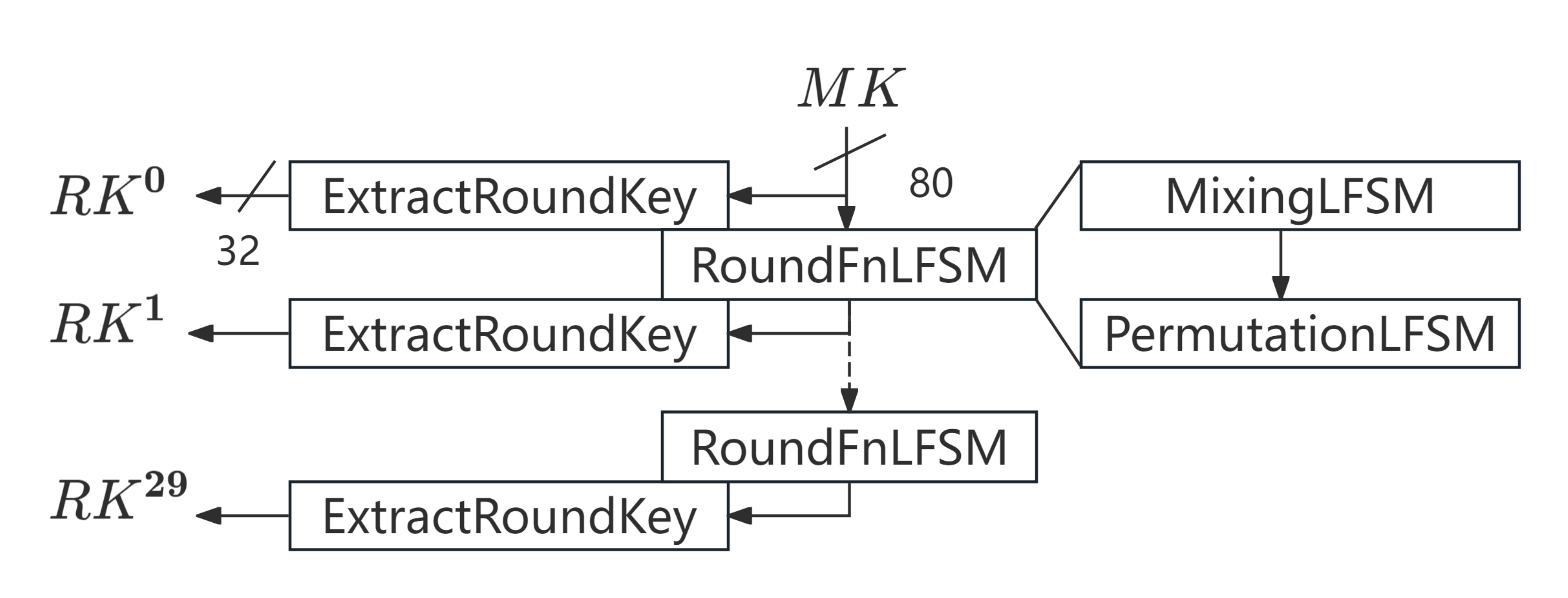}
	\caption{The overview of \cipher's key schedule.}
	\label{fig_3}
\end{figure}

In the LFSM, the 80-bit internal state is divided into 20 nibbles as $K = K_{19}||\cdots||K_0$, where $K_{19}$ is the most significant nibble.
\begin{itemize}
	\item \textsf{ExtractRoundKey}. In each encryption round, the subkey $RK^r$ is extracted from the LFSM internal state using a nonlinear extraction function. First, some nibbles of the state in LFSM are extracted:
	\[
	Z_{(32)} \leftarrow K_{18} \parallel K_{16} \parallel K_{13} \parallel K_{10} \parallel K_9 \parallel K_6 \parallel K_3 \parallel K_1.
	\]
	Denote by $Z_j$ ($0\leq j\leq 31$) the $j$-th bit of $Z_{(32)}$, where $Z_{31}$ is the most significant bit of $Z_{(32)}$. The $i$-th nibble of $RK^r$ is then calculated as
	\[
	RK_i^r = S\left( Z_i \parallel Z_{i+8} \parallel Z_{i+16} \parallel Z_{i+24} \right), \quad 0 \leq i \leq 7,
	\]
	where \( S \) is the S-box defined in Table~\ref{tab:sbox}. Finally, the round constant $r \in \{ 0, \dots, 29 \}$ is XORed to the subkey: $RK^r \leftarrow RK^r \oplus r_{(5)} \parallel 0_{(27)}$.
	
	\item \textsf{RoundFnLFSM}. The update function of LFSM is composed of four parallel operations: $\zeta_0, \zeta_1, \zeta_2$, and $\zeta_3$. Let $K'$ denote the new LFSM state obtained by updating $K$. As shown in Fig.~\ref{fig_4}, $\zeta_i$ is applied to the state $(K_{5*i+4},K_{5*i+3},K_{5*i+2},K_{5*i+1},K_{5*i})$ for $0\leq i\leq 3$. For instance, $\zeta_0$ updates $(K_{4},K_{3},K_{2},K_{1},K_{0})$ as:
	\[
	\begin{cases}
		K'_4 =K_3\\
		K'_3 =K_2\\
		K'_2 =  K_{1}\oplus \left (K_{2}\gg 3\right )\\
		K'_1 = K_{0}\oplus \left (K_{4}\ggg 1\right )\\
		K'_0 = K_4.
	\end{cases}
	\]
\end{itemize}
\begin{figure}[!htp]
	\centering
	\includegraphics[width=0.6\linewidth]{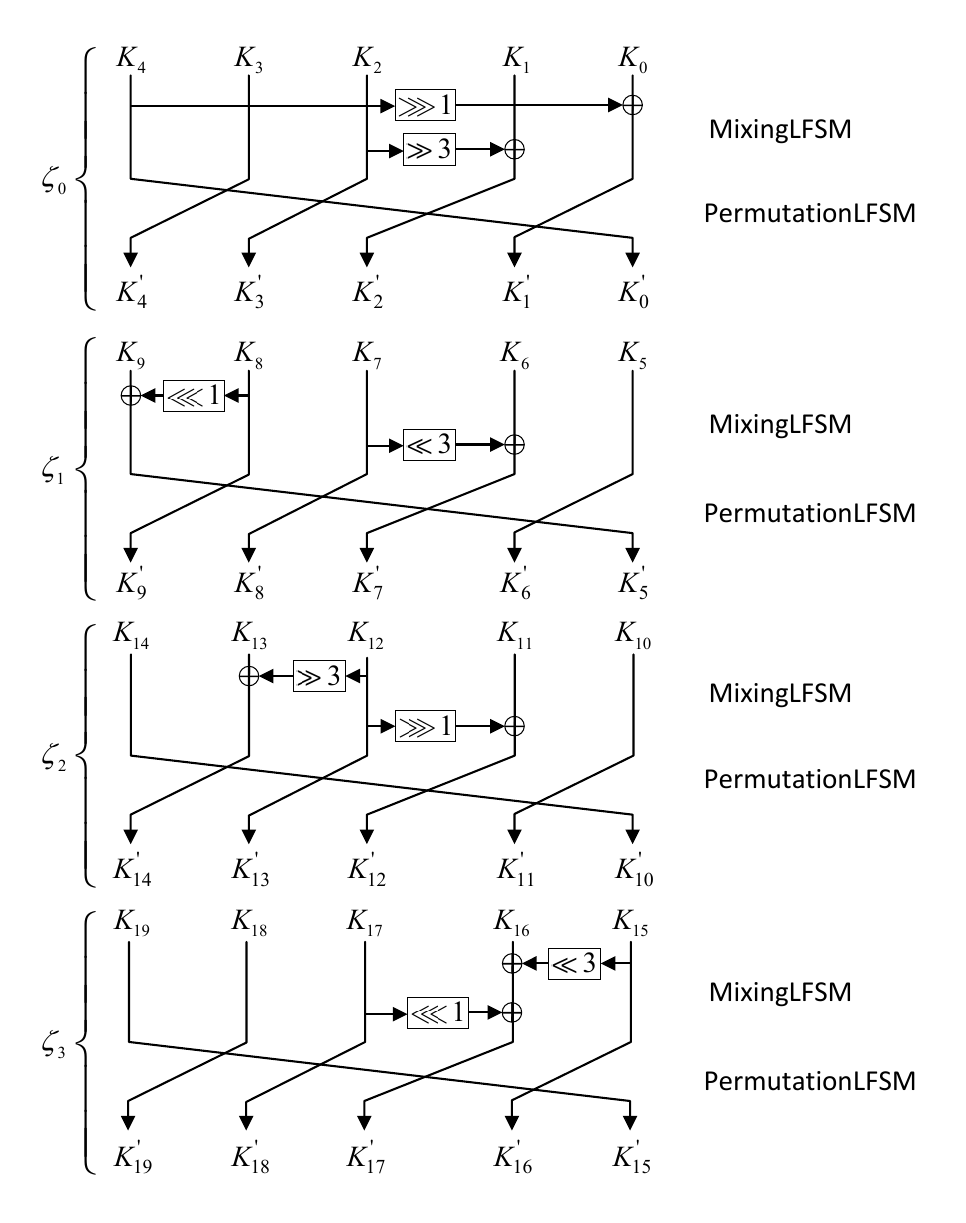}
	\caption{The update function of LFSM.}
	\label{fig_4}
\end{figure}
\section{Raw Experimental Data for the Model}\label{sm:1}

\begin{table}[!htbp]
\centering
\caption{Distribution of the number of injected faults required for successfully recovering subkeys $RK^{29}$ and $RK^{28}$, based on $2^{15}$ experiments}
\label{tab:distribution31}
\setlength{\tabcolsep}{4.5pt} 
\renewcommand{\arraystretch}{1.0} 
\footnotesize 
\begin{tabular}{ccc @{\hspace{2em}} ccc}
\toprule
Fault Count & Frequency & Percentage & Fault Count & Frequency & Percentage \\
\midrule
2  & 16   & 0.05\% & 49  & 83 & 0.25\% \\
3  & 118  & 0.36\% & 50  & 88 & 0.27\% \\
4  & 242  & 0.74\% & 51  & 72 & 0.22\% \\
5  & 370  & 1.13\% & 52  & 66 & 0.20\% \\
6  & 441  & 1.35\% & 53  & 56 & 0.17\% \\
7  & 496  & 1.51\% & 54  & 55 & 0.17\% \\
8  & 590  & 1.80\% & 55  & 33 & 0.10\% \\
9  & 653  & 1.99\% & 56  & 37 & 0.11\% \\
10 & 727  & 2.22\% & 57  & 22 & 0.07\% \\
11 & 787  & 2.40\% & 58  & 32 & 0.10\% \\
12 & 950  & 2.90\% & 59  & 24 & 0.07\% \\
13 & 1022 & 3.12\% & 60  & 23 & 0.07\% \\
14 & 1121 & 3.42\% & 61  & 29 & 0.09\% \\
15 & 1208 & 3.69\% & 62  & 14 & 0.04\% \\
16 & 1268 & 3.87\% & 63  & 20 & 0.06\% \\
17 & 1300 & 3.97\% & 64  & 14 & 0.04\% \\
18 & 1340 & 4.09\% & 65  & 13 & 0.04\% \\
19 & 1358 & 4.14\% & 66  & 14 & 0.04\% \\
20 & 1384 & 4.22\% & 67  & 12 & 0.04\% \\
21 & 1381 & 4.21\% & 68  & 10 & 0.03\% \\
22 & 1313 & 4.01\% & 69  & 11 & 0.03\% \\
23 & 1200 & 3.66\% & 70  & 6  & 0.02\% \\
24 & 1183 & 3.61\% & 71  & 10 & 0.03\% \\
25 & 1115 & 3.40\% & 72  & 8  & 0.02\% \\
26 & 1087 & 3.32\% & 73  & 7  & 0.02\% \\
27 & 930  & 2.84\% & 74  & 3  & 0.01\% \\
28 & 924  & 2.82\% & 75  & 3  & 0.01\% \\
29 & 847  & 2.58\% & 76  & 4  & 0.01\% \\
30 & 779  & 2.38\% & 77  & 3  & 0.01\% \\
31 & 699  & 2.13\% & 78  & 4  & 0.01\% \\
32 & 643  & 1.96\% & 79  & 3  & 0.01\% \\
33 & 575  & 1.75\% & 80  & 4  & 0.01\% \\
34 & 538  & 1.64\% & 82  & 2  & 0.01\% \\
35 & 469  & 1.43\% & 83  & 1  & 0.00\% \\
36 & 415  & 1.27\% & 85  & 1  & 0.00\% \\
37 & 332  & 1.01\% & 86  & 1  & 0.00\% \\
38 & 337  & 1.03\% & 87  & 1  & 0.00\% \\
39 & 292  & 0.89\% & 90  & 1  & 0.00\% \\
40 & 242  & 0.74\% & 91  & 1  & 0.00\% \\
41 & 249  & 0.76\% & 94  & 1  & 0.00\% \\
42 & 175  & 0.53\% & 95  & 1  & 0.00\% \\
43 & 179  & 0.55\% & 96  & 1  & 0.00\% \\
44 & 180  & 0.55\% & 105 & 1  & 0.00\% \\
45 & 141  & 0.43\% & 106 & 2  & 0.01\% \\
46 & 136  & 0.42\% & 111 & 1  & 0.00\% \\
47 & 117  & 0.36\% & 134 & 1  & 0.00\% \\
48 & 91   & 0.28\% &     &    &        \\
\bottomrule
\end{tabular}
\end{table}

\begin{table}[!htbp]
\centering
\caption{Distribution of the number of injected faults required for successfully recovering $RK^{29}$ and $RK^{28}$ excluding $RK^{28}_3$, based on $2^{15}$ experiments}\label{tab:distribution32}%
\label{tab:distribution32}
\setlength{\tabcolsep}{4.5pt}
\renewcommand{\arraystretch}{1.0}
\footnotesize
\begin{tabular}{ccc @{\hspace{2em}} ccc}
\toprule
Fault Count & Frequency & Percentage & Fault Count & Frequency & Percentage \\
\midrule
2  & 16   & 0.05\% & 40 & 122 & 0.37\% \\
3  & 123  & 0.38\% & 41 & 119 & 0.36\% \\
4  & 252  & 0.77\% & 42 & 90  & 0.27\% \\
5  & 383  & 1.17\% & 43 & 92  & 0.28\% \\
6  & 444  & 1.35\% & 44 & 84  & 0.26\% \\
7  & 514  & 1.57\% & 45 & 62  & 0.19\% \\
8  & 656  & 2.00\% & 46 & 49  & 0.15\% \\
9  & 802  & 2.45\% & 47 & 46  & 0.14\% \\
10 & 942  & 2.87\% & 48 & 43  & 0.13\% \\
11 & 1037 & 3.16\% & 49 & 30  & 0.09\% \\
12 & 1240 & 3.78\% & 50 & 32  & 0.10\% \\
13 & 1311 & 4.00\% & 51 & 28  & 0.09\% \\
14 & 1426 & 4.35\% & 52 & 17  & 0.05\% \\
15 & 1542 & 4.71\% & 53 & 9   & 0.03\% \\
16 & 1550 & 4.73\% & 54 & 15  & 0.05\% \\
17 & 1555 & 4.75\% & 55 & 8   & 0.02\% \\
18 & 1599 & 4.88\% & 56 & 8   & 0.02\% \\
19 & 1605 & 4.90\% & 57 & 5   & 0.02\% \\
20 & 1606 & 4.90\% & 58 & 6   & 0.02\% \\
21 & 1502 & 4.58\% & 59 & 2   & 0.01\% \\
22 & 1407 & 4.29\% & 60 & 6   & 0.02\% \\
23 & 1220 & 3.72\% & 61 & 7   & 0.02\% \\
24 & 1201 & 3.67\% & 62 & 4   & 0.01\% \\
25 & 1092 & 3.33\% & 63 & 4   & 0.01\% \\
26 & 1022 & 3.12\% & 64 & 4   & 0.01\% \\
27 & 843  & 2.57\% & 65 & 4   & 0.01\% \\
28 & 794  & 2.42\% & 66 & 3   & 0.01\% \\
29 & 682  & 2.08\% & 67 & 2   & 0.01\% \\
30 & 615  & 1.88\% & 68 & 2   & 0.01\% \\
31 & 500  & 1.53\% & 69 & 1   & 0.00\% \\
32 & 471  & 1.44\% & 70 & 2   & 0.01\% \\
33 & 397  & 1.21\% & 71 & 1   & 0.00\% \\
34 & 369  & 1.13\% & 73 & 2   & 0.01\% \\
35 & 281  & 0.86\% & 74 & 2   & 0.01\% \\
36 & 263  & 0.80\% & 75 & 1   & 0.00\% \\
37 & 225  & 0.69\% & 76 & 1   & 0.00\% \\
38 & 190  & 0.58\% &    &     &        \\
39 & 170  & 0.52\% &    &     &        \\
\bottomrule
\end{tabular}
\end{table}

\begin{table}[!htbp]
\centering
\caption{Key-recovery success rates for different numbers of injected faults, based on $2^{15}$ experiments}
\label{tab:successs_rate3}
\setlength{\tabcolsep}{6pt}
\renewcommand{\arraystretch}{1.02}
\footnotesize
\begin{tabular}{cc @{\hspace{2.5em}} cc}
\toprule
Fault Count & Success Rate & Fault Count & Success Rate \\
\midrule
2  & 0.05\%  & 46 & 99.02\% \\
3  & 0.46\%  & 47 & 99.25\% \\
4  & 1.17\%  & 48 & 99.28\% \\
5  & 2.22\%  & 49 & 99.33\% \\
6  & 3.62\%  & 50 & 99.47\% \\
7  & 5.44\%  & 51 & 99.52\% \\
8  & 7.08\%  & 52 & 99.56\% \\
9  & 9.87\%  & 53 & 99.66\% \\
10 & 12.98\% & 54 & 99.69\% \\
11 & 16.17\% & 55 & 99.67\% \\
12 & 19.87\% & 56 & 99.74\% \\
13 & 24.33\% & 57 & 99.75\% \\
14 & 28.65\% & 58 & 99.80\% \\
15 & 33.18\% & 59 & 99.76\% \\
16 & 38.35\% & 60 & 99.82\% \\
17 & 42.69\% & 61 & 99.82\% \\
18 & 47.49\% & 62 & 99.84\% \\
19 & 52.71\% & 63 & 99.86\% \\
20 & 56.67\% & 64 & 99.89\% \\
21 & 61.25\% & 65 & 99.85\% \\
22 & 64.97\% & 66 & 99.86\% \\
23 & 69.38\% & 67 & 99.86\% \\
24 & 73.08\% & 68 & 99.84\% \\
25 & 76.44\% & 69 & 99.89\% \\
26 & 79.02\% & 70 & 99.82\% \\
27 & 81.74\% & 71 & 99.87\% \\
28 & 83.93\% & 72 & 99.84\% \\
29 & 86.08\% & 73 & 99.86\% \\
30 & 88.18\% & 74 & 99.85\% \\
31 & 89.70\% & 75 & 99.87\% \\
32 & 91.15\% & 76 & 99.86\% \\
33 & 92.62\% & 77 & 99.84\% \\
34 & 93.65\% & 78 & 99.82\% \\
35 & 94.64\% & 79 & 99.87\% \\
36 & 95.24\% & 80 & 99.84\% \\
37 & 95.88\% & 81 & 99.87\% \\
38 & 96.36\% & 82 & 99.85\% \\
39 & 96.92\% & 83 & 99.84\% \\
40 & 97.45\% & 84 & 99.87\% \\
41 & 98.07\% & 85 & 99.87\% \\
42 & 98.14\% & 86 & 99.87\% \\
43 & 98.45\% & 87 & 99.89\% \\
44 & 98.71\% & 88 & 99.86\% \\
45 & 98.93\% & 89 & 99.87\% \\
\bottomrule
\end{tabular}
\end{table}
\clearpage 

\section{Key-Recovery Performance for a Fixed Fault Location}\label{sm:fixed_x7}
The fixed-location fault model can be regarded as a special case of the random-location fault model considered in this paper. Specifically, we consider the case where faults are consistently injected into branch $X_7^{27}$. Since the corresponding differential propagation and key-recovery procedures have already been described in \ref{sect:4.3}, we focus here only on the experimental results.

Fig.~\ref{fig:fault_distribution_fixed_x7} shows the distribution of the number of faulty ciphertexts required for successful key recovery, while Fig.~\ref{fig:success_rate_fixed_x7} presents the key-recovery success rate for different numbers of faulty ciphertexts. These results provide a fixed-location reference for evaluating the key-recovery performance of the proposed attack under the more general random-location setting.
\begin{figure}[!b]
	\centering
	\includegraphics[width=0.66\linewidth]{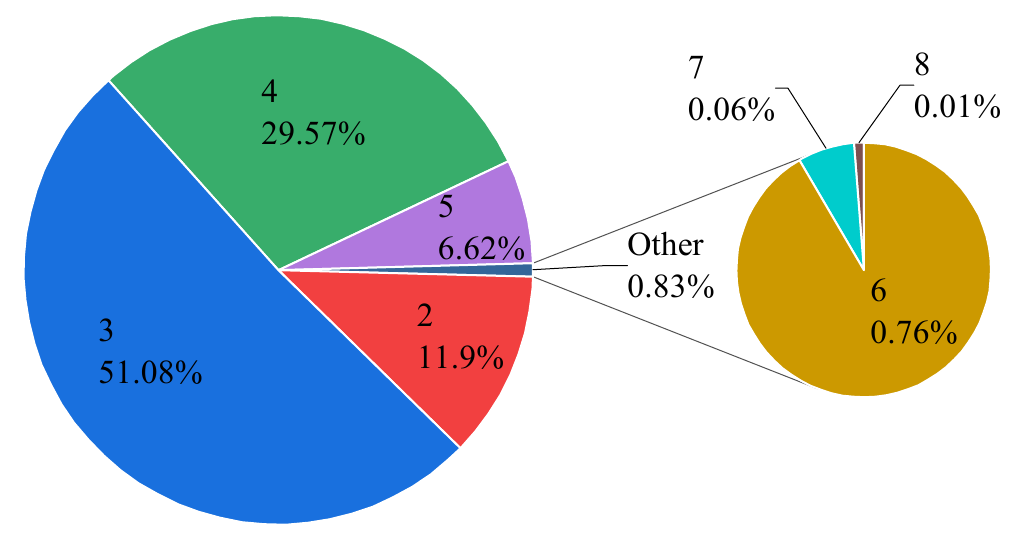}
	\caption{Distribution of the number of faults injected at the fixed location $X_7^{27}$ required to recover $RK^{29}$ and $RK^{28}$, based on $2^{15}$ experiments. For exact statistical values, please refer to Table~\ref{tab:fault_distribution_fixed_x7}.}
	\label{fig:fault_distribution_fixed_x7}
\end{figure}
\begin{figure}[!b]
	\centering
	\includegraphics[width=0.8\linewidth]{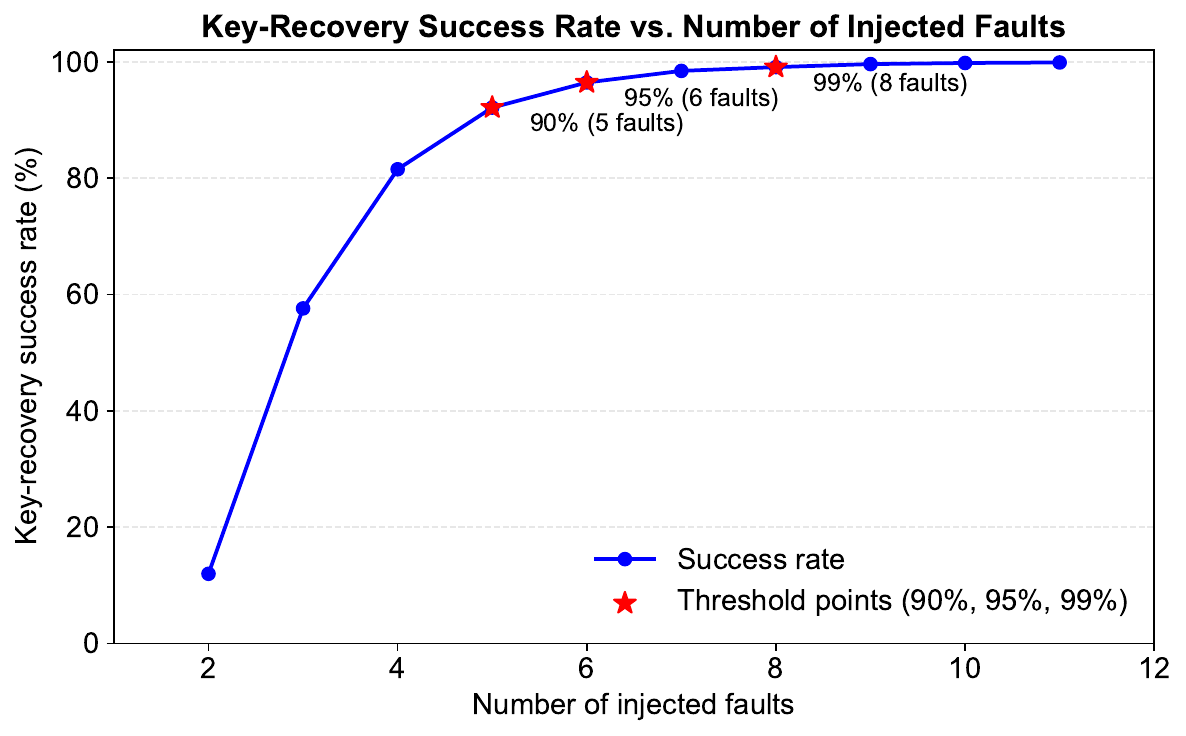}
	\caption{Key-recovery success rates for different numbers of faults injected at the fixed location $X_7^{27}$. For exact statistical values, please refer to Table~\ref{tab:success_rate_fixed_x7}.}
	\label{fig:success_rate_fixed_x7}
\end{figure}

\begin{table}[h]
\caption{Distribution of the number of faults injected at the fixed location $X_7^{27}$ required to recover $RK^{29}$ and $RK^{28}$, based on $2^{15}$ experiments.}
\label{tab:fault_distribution_fixed_x7}
\begin{tabular*}{\textwidth}{@{\extracolsep\fill}ccccccccc}
\toprule
& \multicolumn{8}{@{}c@{}}{Number of Experiments} \\
\cmidrule{2-9}
Fault Number & $2^{8}$ & $2^{9}$ & $2^{10}$ & $2^{11}$ & $2^{12}$ & $2^{13}$ & $2^{14}$ & $2^{15}$ \\
\midrule
 2 & 15.62\% & 11.91\% & 11.04\% & 11.43\% & 10.91\% & 11.28\% & 11.92\% & 11.90\% \\
3 & 46.48\% & 52.34\% & 52.34\% & 49.41\% & 51.56\% & 51.61\% & 50.60\% & 51.08\% \\
		
4 & 31.64\% & 26.95\% & 28.32\% & 30.81\% & 30.18\% & 29.81\% & 30.10\% & 29.57\% \\
	
5 &  6.25\% &  7.03\% &  7.91\% &  7.42\% &  6.40\% &  6.41\% &  6.65\% &  6.62\% \\
		
6 &  0      &  1.76\% &  0.29\% &  0.88\% &  0.85\% &  0.83\% &  0.66\% &  0.76\% \\
	
7 &  0      &  0      &  0.10\% &  0.05\% &  0.10\% &  0.05\% &  0.07\% &  0.06\% \\
		
8 &  0      &  0      &  0      &  0      &  0      &  0.01\% &  0.01\% &  0.01\% \\
\hline
\end{tabular*}
\end{table}

\begin{table}[h]
\caption{Key-recovery success rates for different numbers of faults injected at the fixed location $X_7^{27}$}
\label{tab:success_rate_fixed_x7}
\begin{tabular*}{\textwidth}{@{\extracolsep\fill}ccccccccc}
\toprule
& \multicolumn{8}{@{}c@{}}{Number of Experiments} \\
\cmidrule{2-9}
Fault Number & $2^{8}$ & $2^{9}$ & $2^{10}$ & $2^{11}$ & $2^{12}$ & $2^{13}$ & $2^{14}$ & $2^{15}$ \\
\midrule
 2  & 12.11\% & 13.48\% & 11.33\% & 12.40\% & 10.45\% & 11.91\% & 11.85\% & 11.96\% \\
3  & 58.59\% & 56.25\% & 61.13\% & 56.15\% & 56.49\% & 57.03\% & 57.60\% & 57.60\% \\
4  & 81.25\% & 78.52\% & 80.96\% & 82.76\% & 82.03\% & 82.39\% & 81.48\% & 81.53\% \\
5  & 92.58\% & 92.77\% & 93.55\% & 92.04\% & 92.29\% & 92.32\% & 92.40\% & 92.14\% \\
6  & 96.48\% & 96.09\% & 97.56\% & 96.34\% & 97.09\% & 96.35\% & 96.52\% & 96.46\% \\
7  & 98.44\% & 98.24\% & 98.34\% & 98.10\% & 98.36\% & 98.38\% & 98.27\% & 98.45\% \\
8  & 99.61\% & 99.61\% & 99.32\% & 99.07\% & 98.97\% & 99.08\% & 99.21\% & 99.10\% \\
9  &100.00\% & 99.22\% & 99.41\% & 99.66\% & 99.68\% & 99.67\% & 99.63\% & 99.62\% \\
10 &100.00\% & 99.80\% & 99.51\% & 99.90\% & 99.83\% & 99.80\% & 99.83\% & 99.81\% \\
11 &100.00\% &100.00\% & 99.90\% & 99.95\% & 99.95\% & 99.83\% & 99.88\% & 99.90\% \\
\hline
\end{tabular*}
\end{table}

\section{Fault Propagation Paths for Injections at Branches $X^{27}_{1}$--$X^{27}_{6}$ and $X^{27}_{8}$--$X^{27}_{15}$}\label{sm:2}
\begin{figure}[!b]
    \centering
    \begin{minipage}{0.49\textwidth}
        \centering
        \includegraphics[width=\linewidth]{fig/differential_3r_x1.pdf}
        \caption{Fault propagation for fault injection at $X^{27}_1$.}
        \label{fig:fig16}
    \end{minipage}
    \hfill
    \begin{minipage}{0.49\textwidth}
        \centering
        \includegraphics[width=\linewidth]{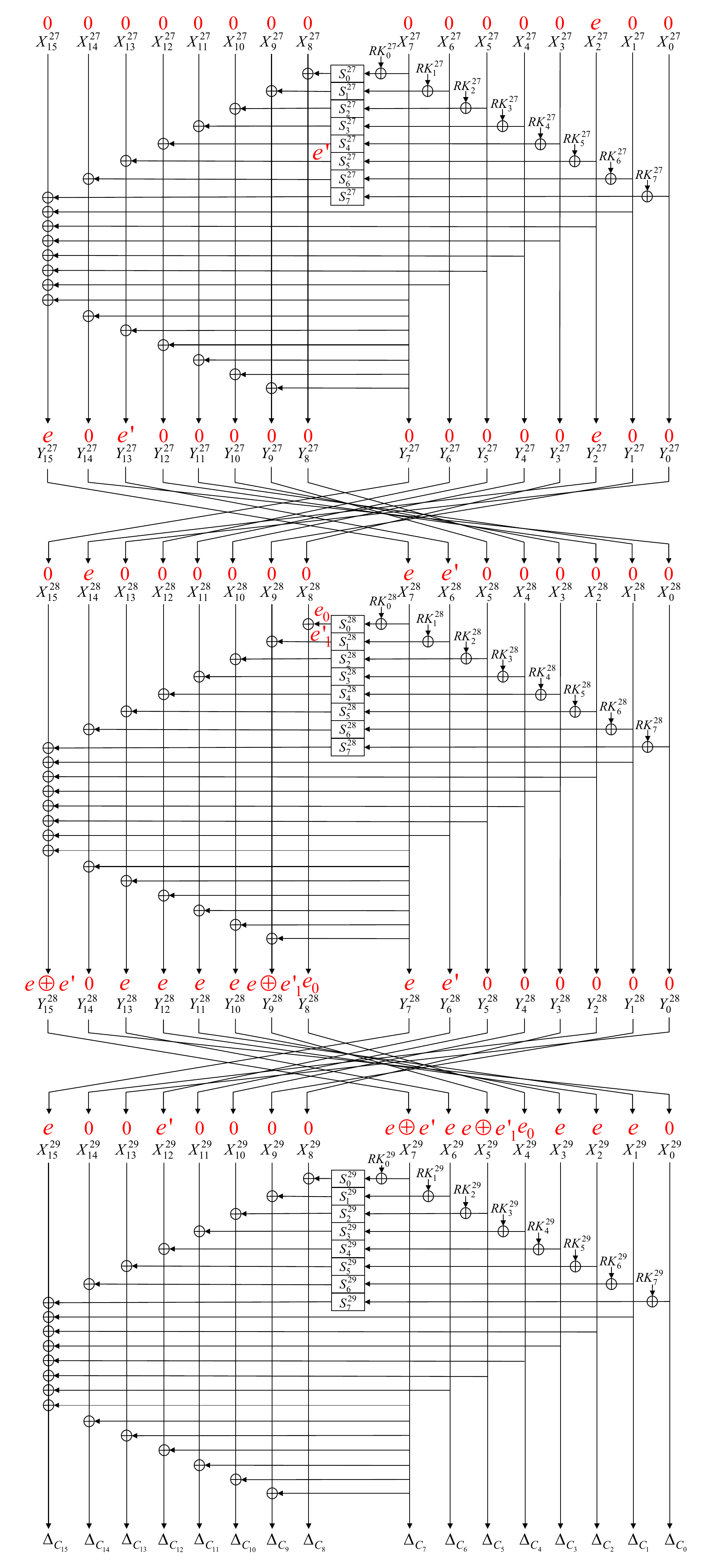}
        \caption{Fault propagation for fault injection at $X^{27}_2$.}
        \label{fig:fig17}
    \end{minipage}
\end{figure}
\begin{figure}[!b]
    \centering
    \begin{minipage}{0.49\textwidth}
        \centering
        \includegraphics[width=\linewidth]{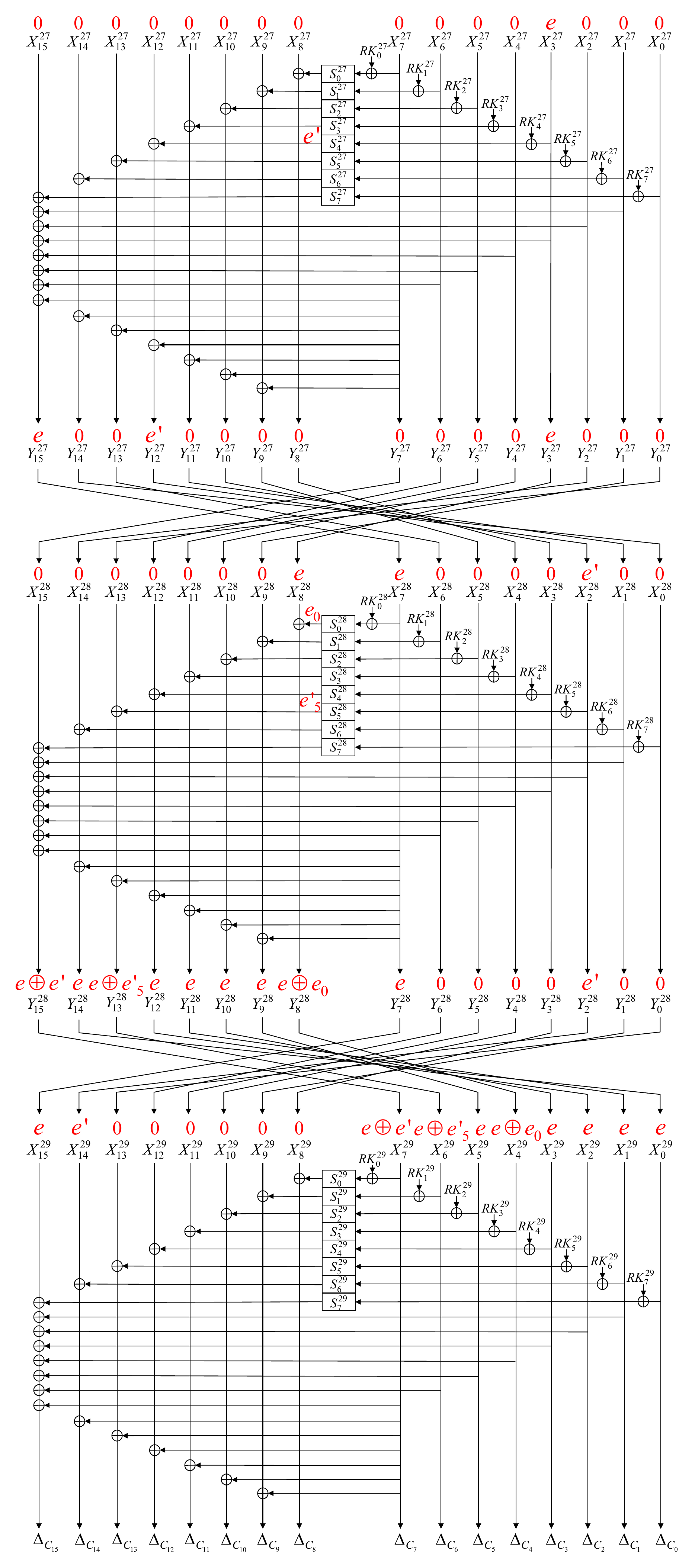}
        \caption{Fault propagation for fault injection at $X^{27}_3$.}
        \label{fig:fig18}
    \end{minipage}
    \hfill
    \begin{minipage}{0.49\textwidth}
        \centering
        \includegraphics[width=\linewidth]{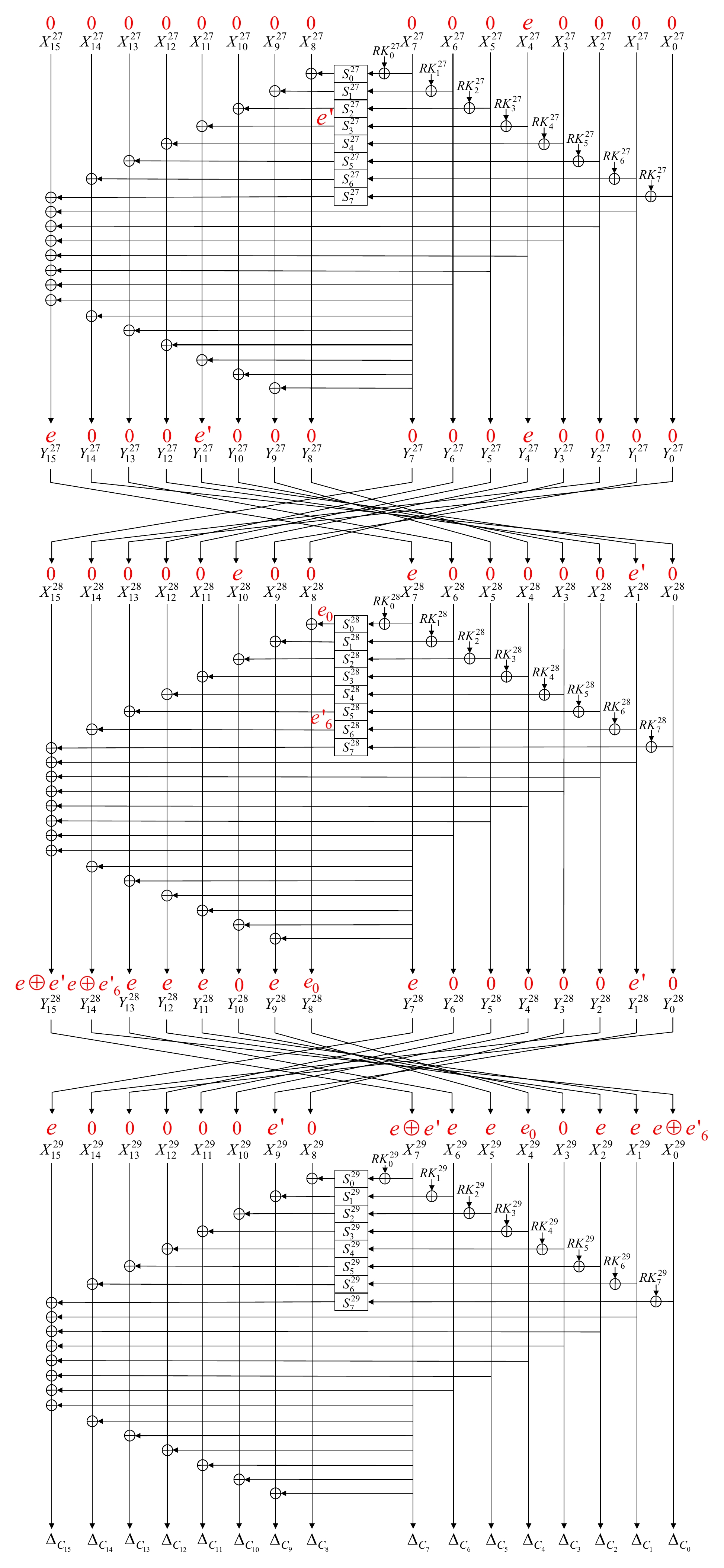}
        \caption{Fault propagation for fault injection at $X^{27}_4$.}
        \label{fig:fig19}
    \end{minipage}
\end{figure}
\begin{figure}[!b]
    \centering
    \begin{minipage}{0.49\textwidth}
        \centering
        \includegraphics[width=\linewidth]{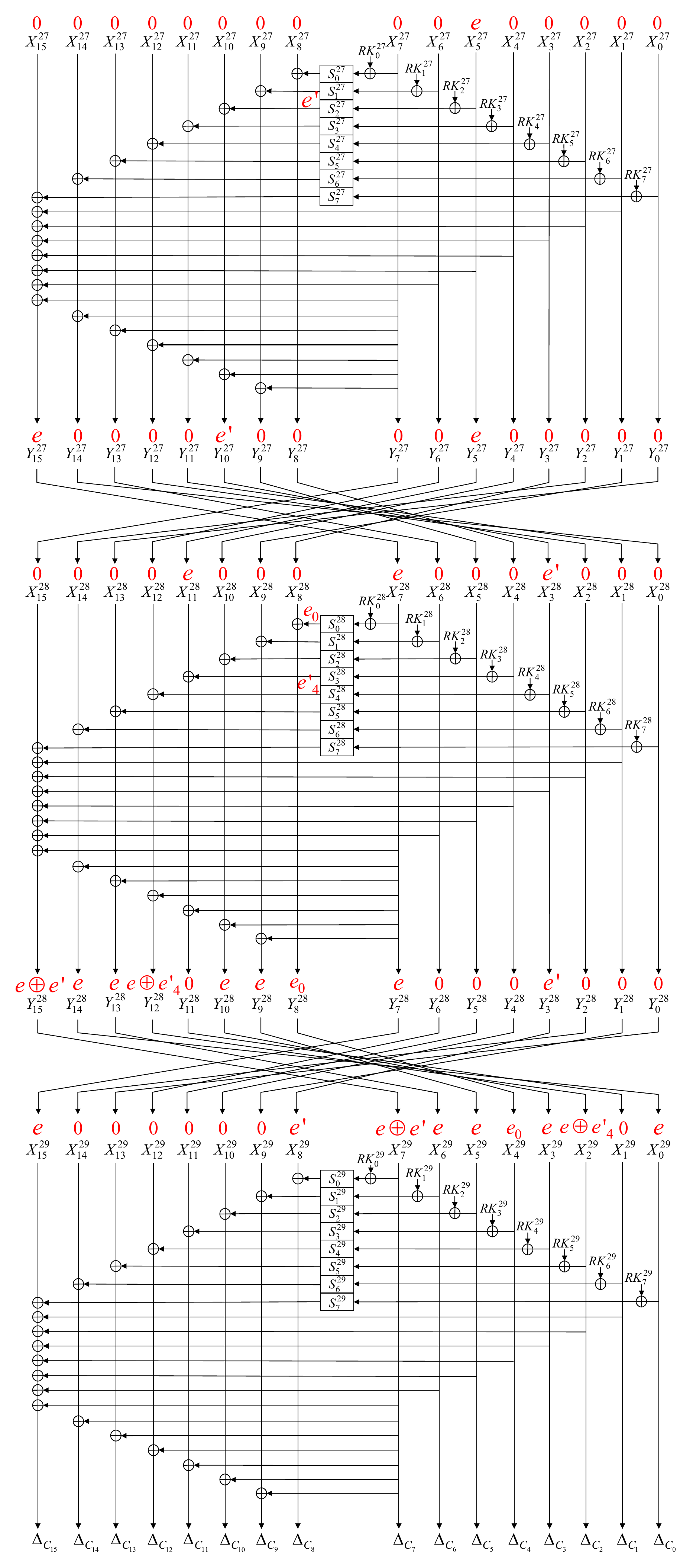}
        \caption{Fault propagation for fault injection at $X^{27}_5$.}
        \label{fig:fig20}
    \end{minipage}
    \hfill
    \begin{minipage}{0.49\textwidth}
        \centering
        \includegraphics[width=\linewidth]{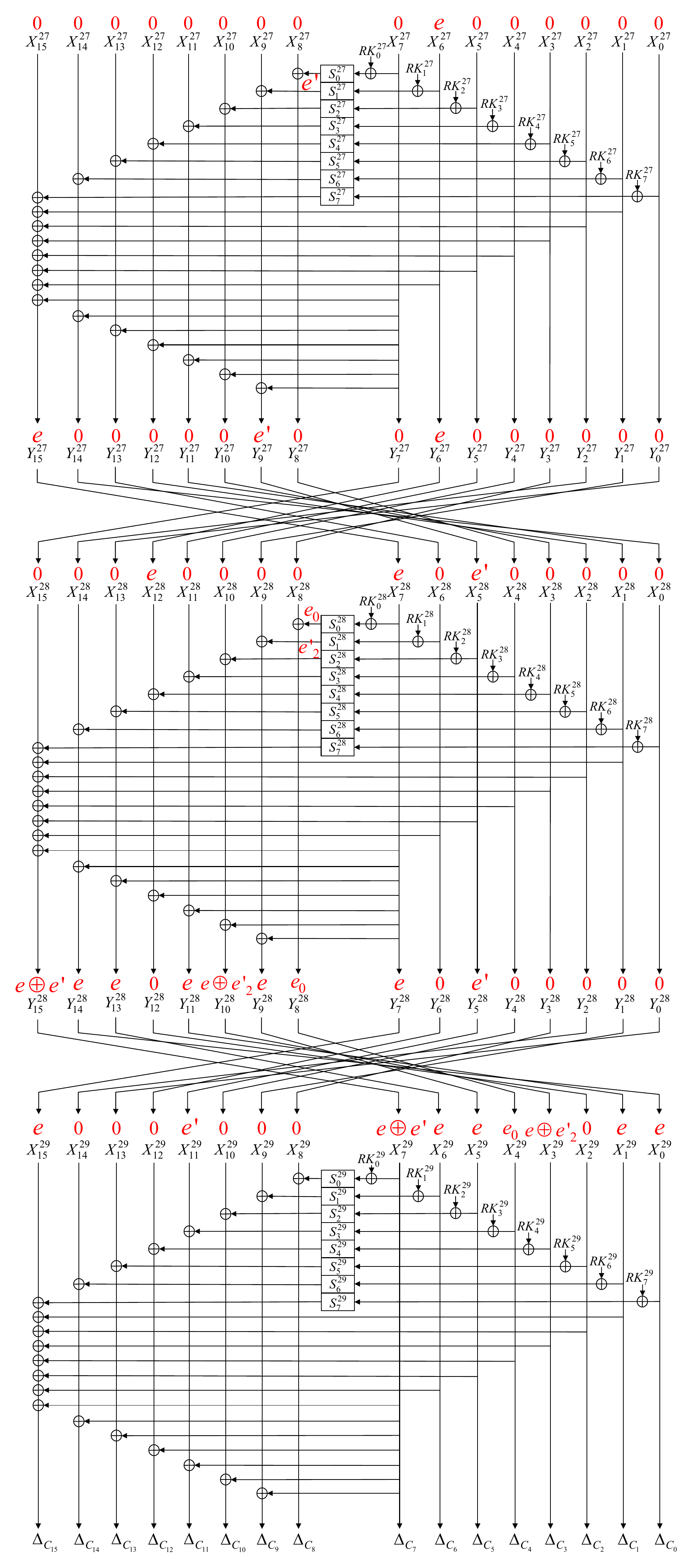}
        \caption{Fault propagation for fault injection at $X^{27}_6$.}
        \label{fig:fig21}
    \end{minipage}
\end{figure}
\begin{figure}[!b]
    \centering
    \begin{minipage}{0.49\textwidth}
        \centering
        \includegraphics[width=\linewidth]{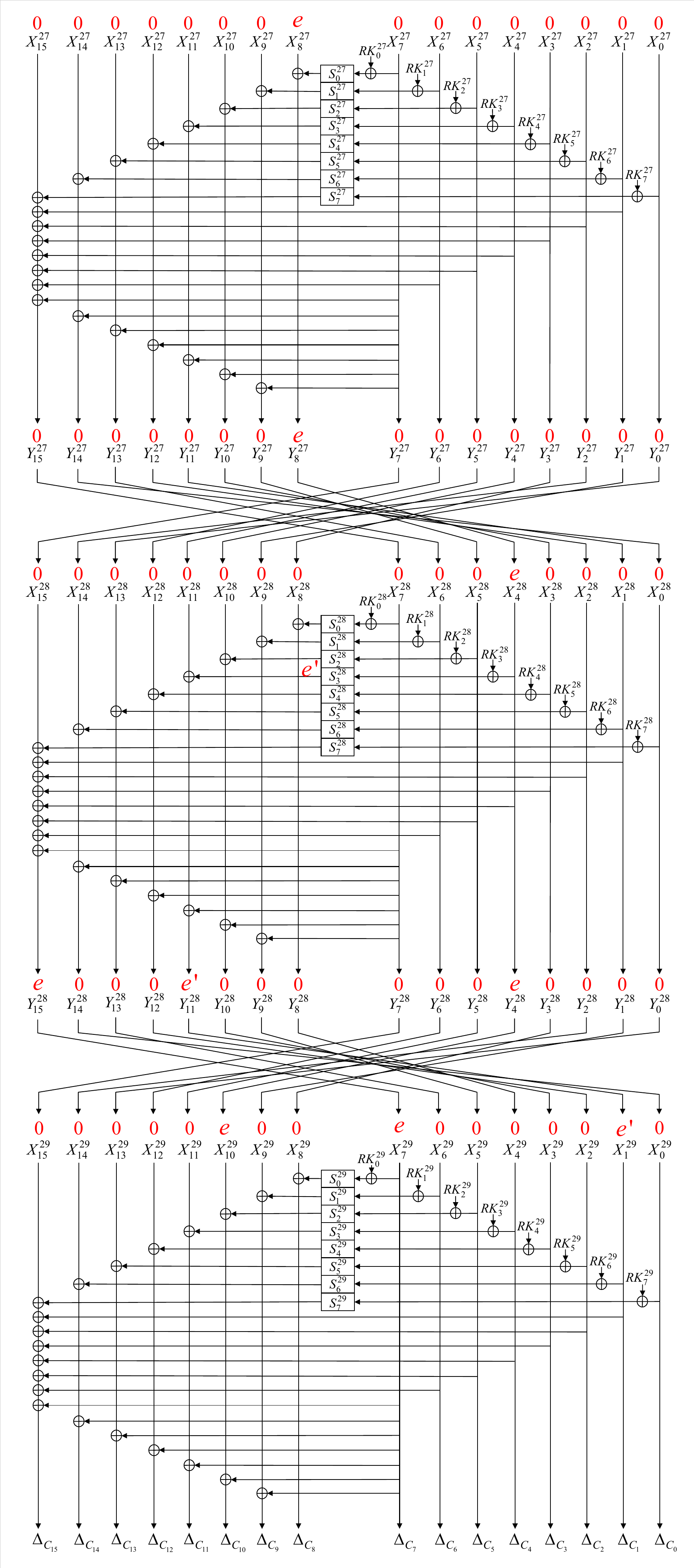}
        \caption{Fault propagation for fault injection at $X^{27}_8$.}
        \label{fig:fig22}
    \end{minipage}
    \hfill
    \begin{minipage}{0.49\textwidth}
        \centering
        \includegraphics[width=\linewidth]{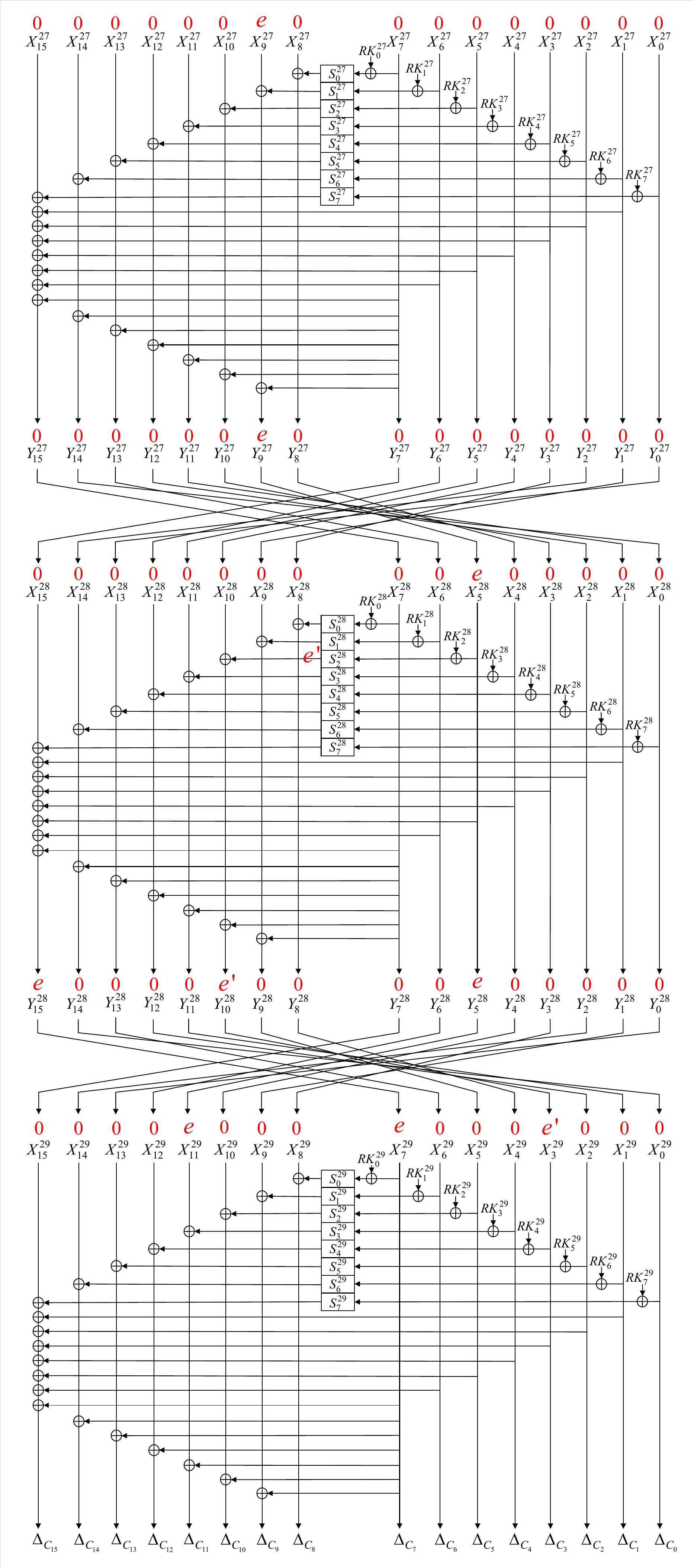}
        \caption{Fault propagation for fault injection at $X^{27}_9$.}
        \label{fig:fig23}
    \end{minipage}
\end{figure}
\begin{figure}[!b]
    \centering
    \begin{minipage}{0.49\textwidth}
        \centering
        \includegraphics[width=\linewidth]{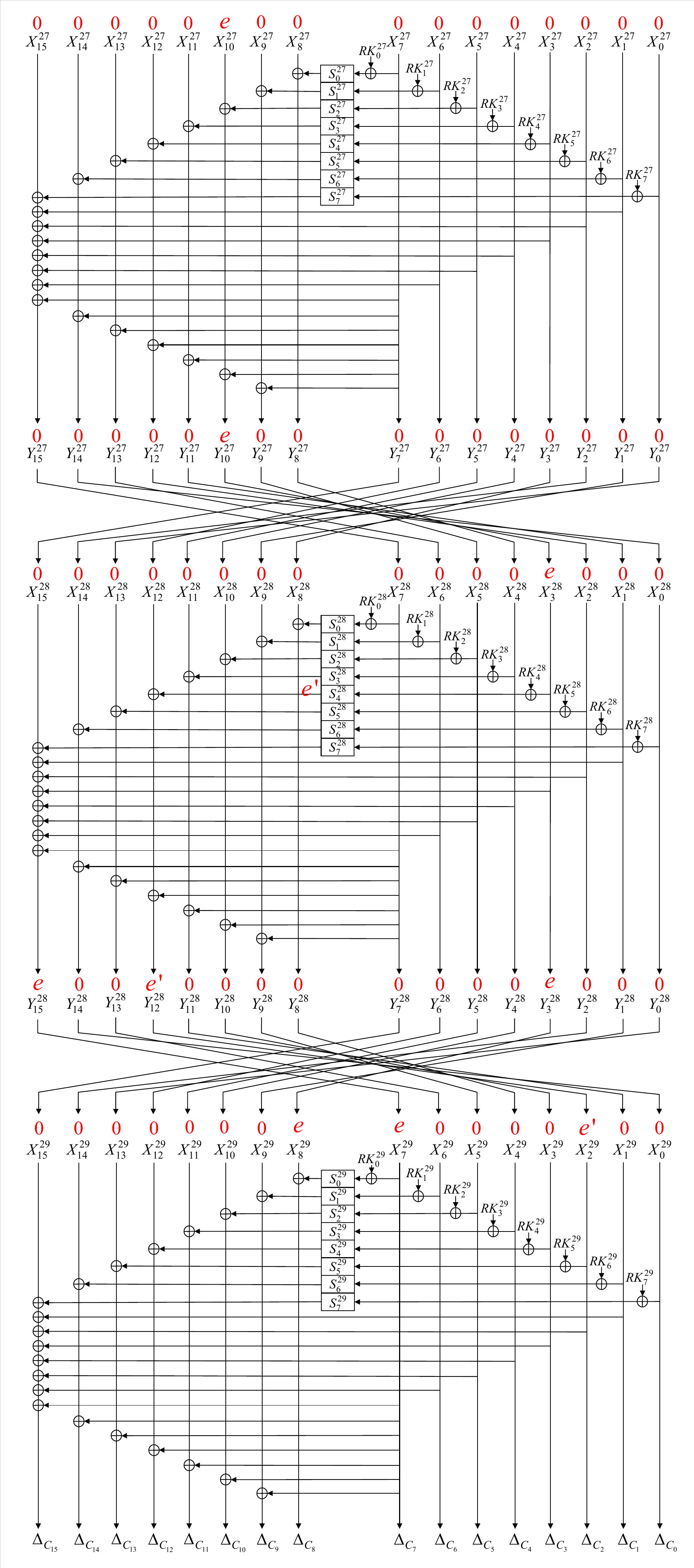}
        \caption{Fault propagation for fault injection at $X^{27}_{10}$.}
        \label{fig:fig24}
    \end{minipage}
    \hfill
    \begin{minipage}{0.49\textwidth}
        \centering
        \includegraphics[width=\linewidth]{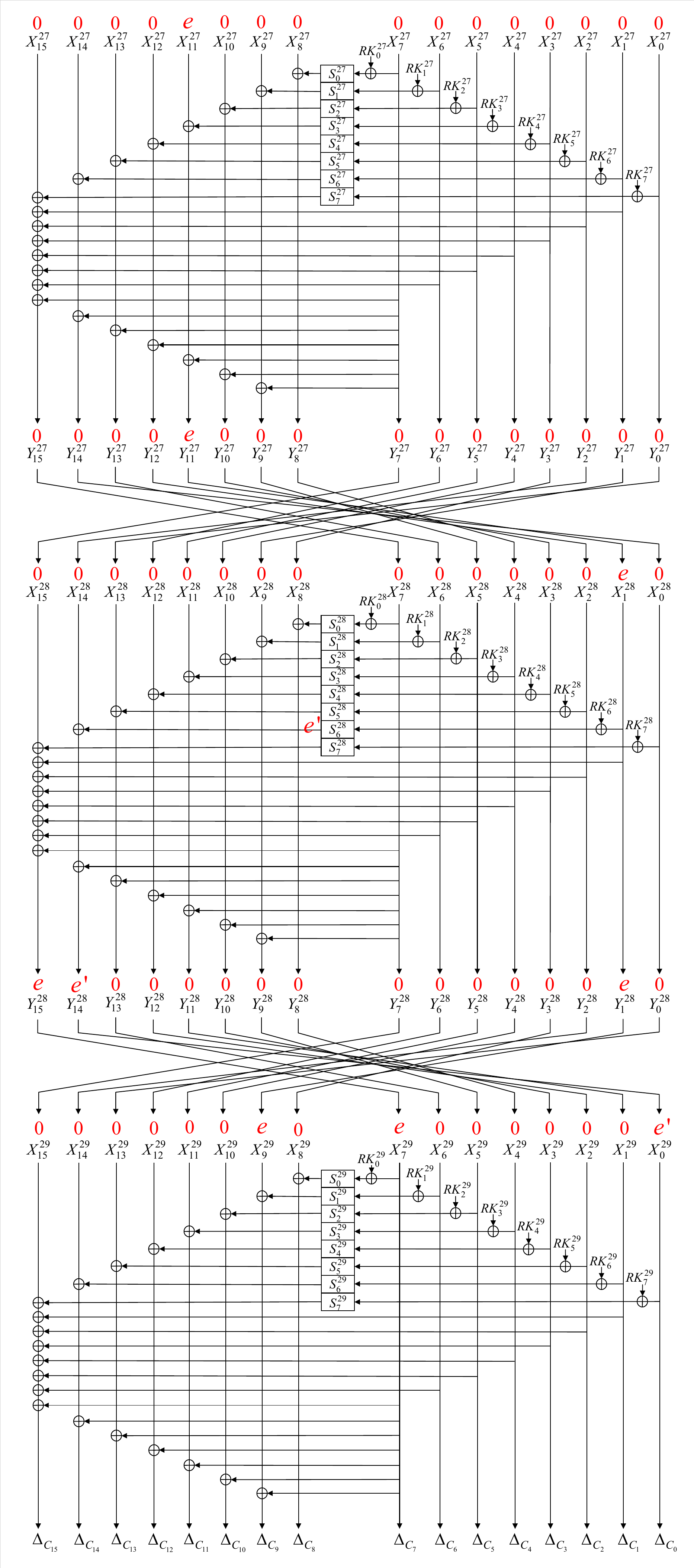}
        \caption{Fault propagation for fault injection at $X^{27}_{11}$.}
        \label{fig:fig25}
    \end{minipage}
\end{figure}
\begin{figure}[!b]
    \centering
    \begin{minipage}{0.49\textwidth}
        \centering
        \includegraphics[width=\linewidth]{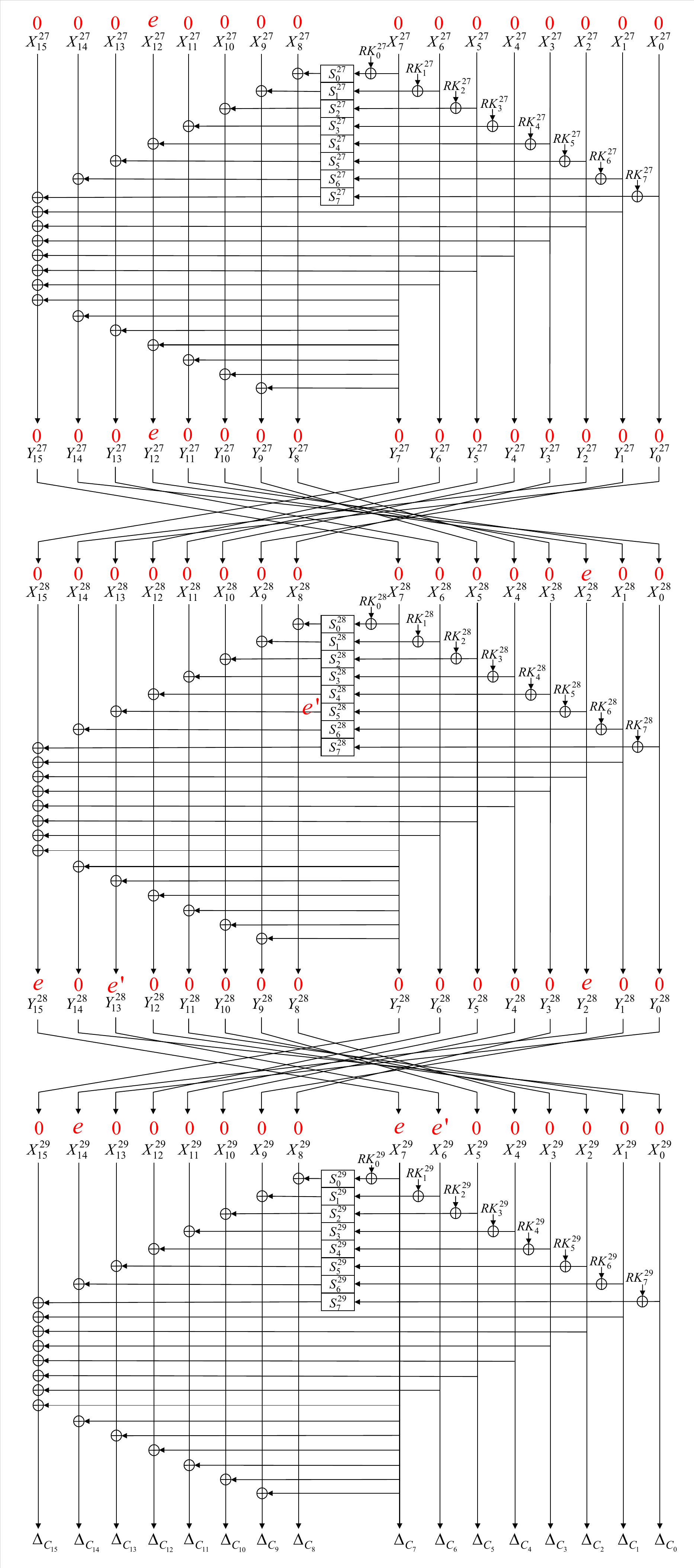}
        \caption{Fault propagation for fault injection at $X^{27}_{12}$.}
        \label{fig:fig26}
    \end{minipage}
    \hfill
    \begin{minipage}{0.49\textwidth}
        \centering
        \includegraphics[width=\linewidth]{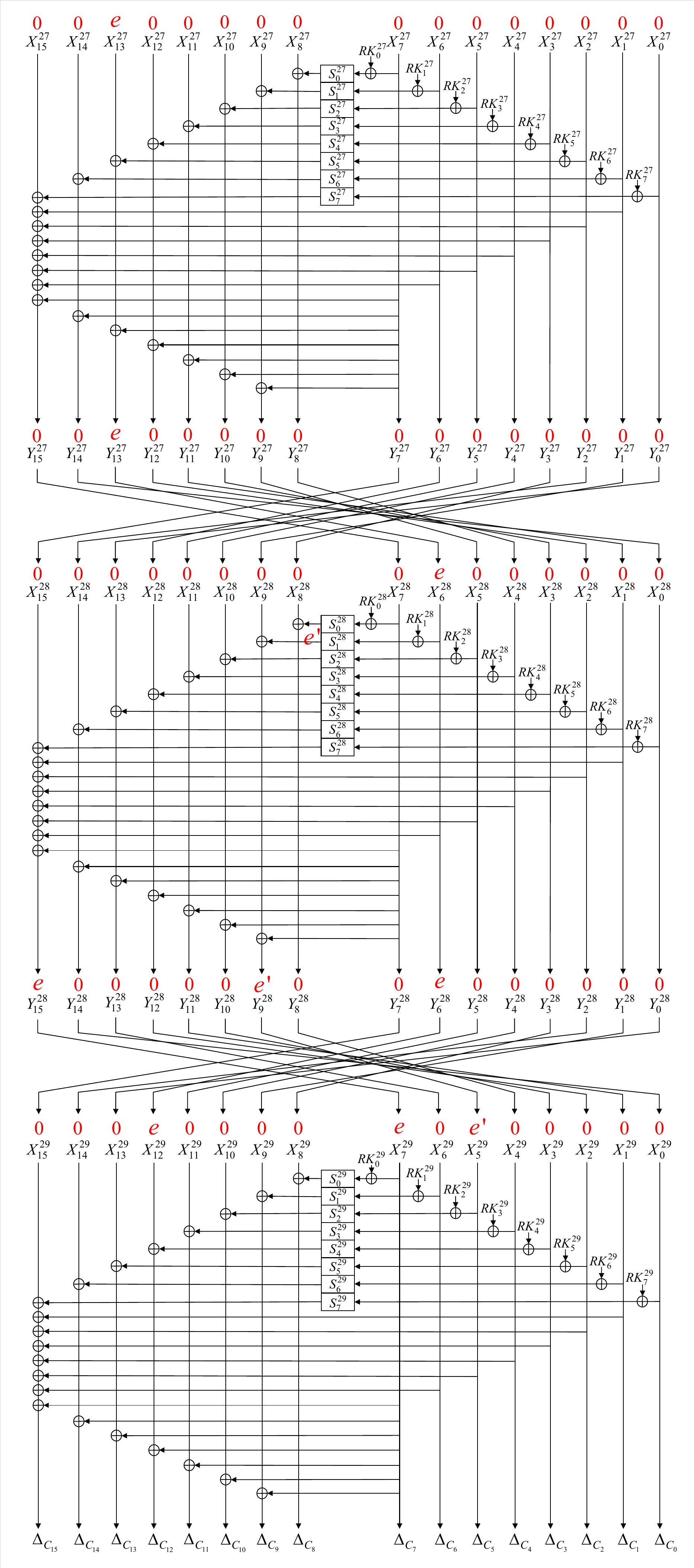}
        \caption{Fault propagation for fault injection at $X^{27}_{13}$.}
        \label{fig:fig27}
    \end{minipage}
\end{figure}
\begin{figure}[!b]
    \centering
    \begin{minipage}{0.49\textwidth}
        \centering
        \includegraphics[width=\linewidth]{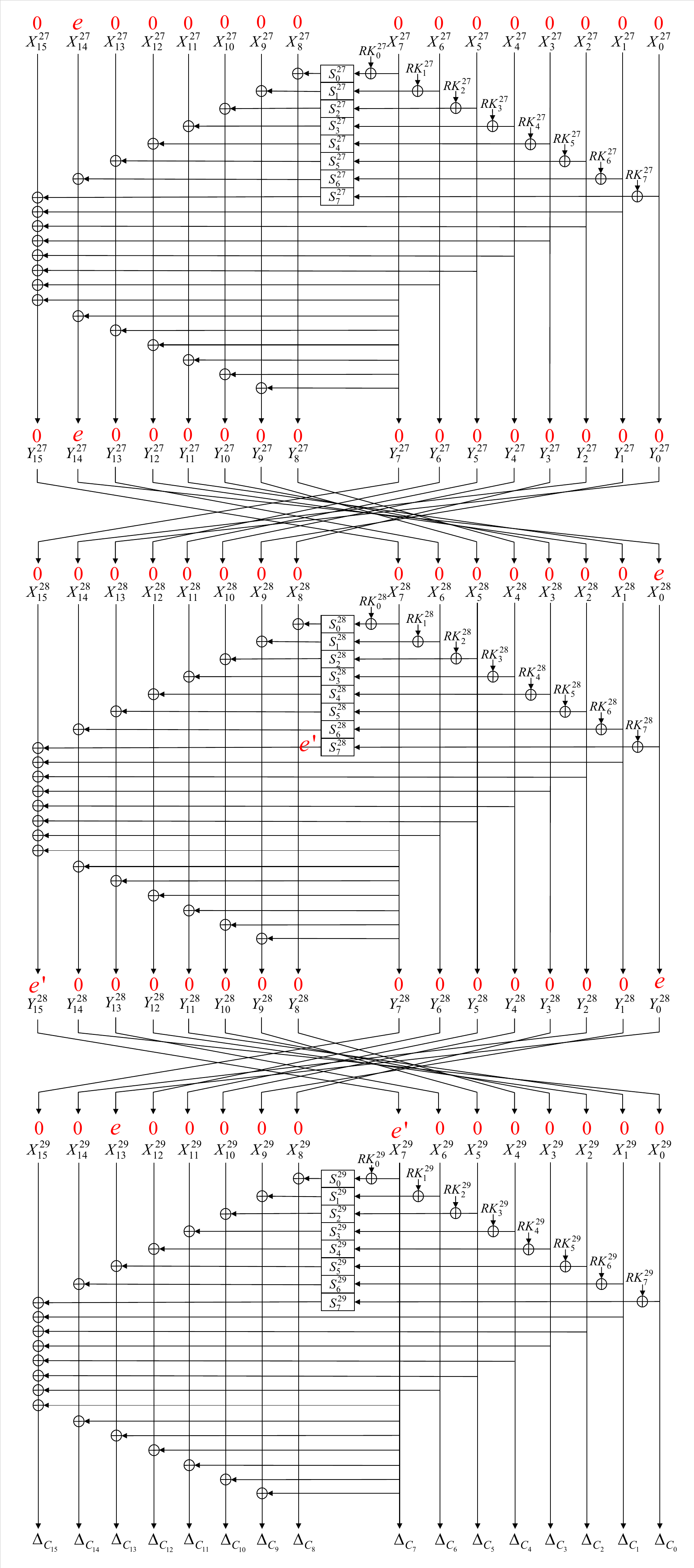}
        \caption{Fault propagation for fault injection at $X^{27}_{14}$.}
        \label{fig:fig28}
    \end{minipage}
    \hfill
    \begin{minipage}{0.49\textwidth}
        \centering
        \includegraphics[width=\linewidth]{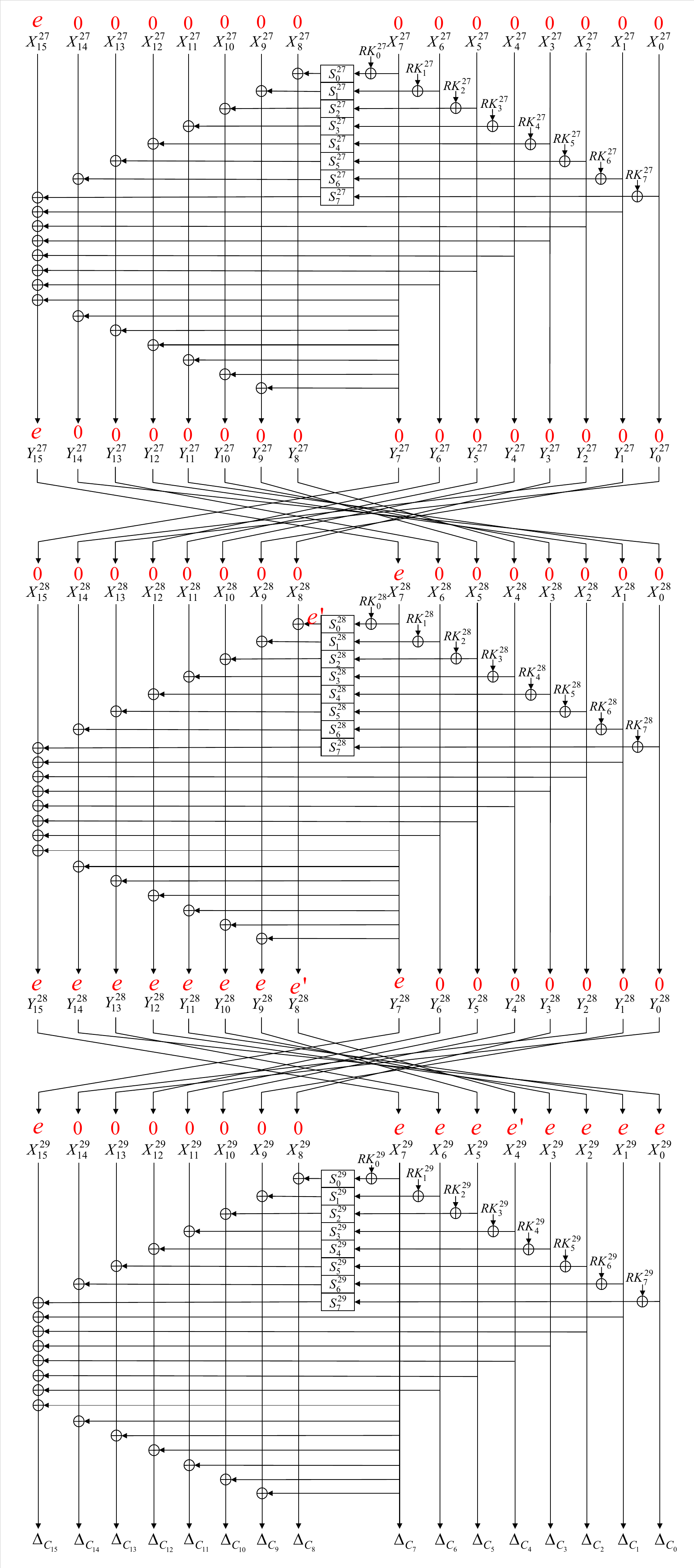}
        \caption{Fault propagation for fault injection at $X^{27}_{15}$.}
        \label{fig:fig29}
    \end{minipage}
\end{figure}
\end{document}